\documentclass[9.5pt, journal]{IEEEtran}
\usepackage{ifthen}
\usepackage{amsmath,bm}
\usepackage{soul}
\IEEEoverridecommandlockouts                              
\usepackage[dvipsnames]{xcolor}
\usepackage{url}
\usepackage{csquotes}

\usepackage{amsfonts,amssymb} 
\usepackage{mathtools,lipsum, nccmath}
\usepackage{graphics} 
\usepackage{mathptmx} 
\usepackage{times} 
\usepackage{cancel}
\usepackage{mathtools}

\usepackage{algorithm}
\usepackage{algorithmic}
\usepackage{lingmacros}
\usepackage{color}
\usepackage{float}
\usepackage{hyperref}
\usepackage{cleveref}
\usepackage{bm}
\usepackage{array}
\usepackage{tabu}

\hypersetup{
    colorlinks=black,
    linkcolor=black,
    filecolor=black,      
    urlcolor=black,
    citecolor=black
    }
	
\usepackage{amsthm}

  {
      \newtheorem{assumption}{Assumption}[section]

      \newtheorem{theorem}{Theorem}[section]
    
    \newtheorem{lemma}{Lemma}[section]
    \newtheorem{remark}{Remark}[section]
    
  }
\let\oldproofname=\proofname
\renewcommand{\proofname}{\rm\bf{\oldproofname}}
\makeatletter
\renewenvironment{proof}[1][\proofname]{%
   \par\pushQED{\qed}\normalfont%
   \topsep6\p@\@plus6\p@\relax
   \trivlist\item[\hskip\labelsep\bfseries#1\@addpunct{.}]%
   \ignorespaces
}{%
   \popQED\endtrivlist\@endpefalse
}
\makeatother

\DeclareMathOperator{\Tr}{Tr}
\DeclareMathOperator{\col}{col}

\newcommand{\mbs}{\mathbf{s}}
\newcommand{\mba}{\mathbf{a}}
\newcommand{\mbf}{\mathbf{f}}
\newcommand{\mbe}{\mathbf{e}}
\newcommand{\mbF}{\mathbf{F}}

\newcommand{\mbA}{\mathbf{A}}
\newcommand{\mbSigma}{\pmb{\Sigma}}
\newcommand{\mbLambda}{\pmb{\Lambda}}
\newcommand{\mbsigma}{\pmb{\sigma}}
\newcommand{\mbmu}{\pmb{\mu}}
\newcommand{\mbepsilon}{\pmb{\epsilon}}
\newcommand{\mbzeta}{\pmb{\zeta}}
\newcommand{\mbz}{\mathbf{z}}

\newcommand{\BoxD}{\textsf{Box2D}}

\usepackage{soul}
\let\oldmb\mathbold
\protected\def\mathbold{\oldmb}

\newcommand{\bluecolor}{\color{blue}}

\usepackage{xcolor}
\usepackage{graphicx}
\usepackage{verbatim}   
\usepackage{color}
\usepackage{cancel}
\usepackage[normalem]{ulem}
\usepackage{subfigure}  
\newcommand*{\rom}[1]{\expandafter\@slowromancap\romannumeral #1@}
\makeatother
\begin{document}
\newcounter{keep_kde_plots}
\setcounter{keep_kde_plots}{1}
%


\title{Stochastic Optimal Control for Multivariable
Dynamical Systems Using Expectation Maximization}

\author{Prakash Mallick and Zhiyong Chen
\thanks{The authors are with the School of Electrical Engineering and Computing,
        University of Newcastle, Callaghan, NSW 2308, Australia. Emails: {\tt\small Prakash.Mallick@uon.edu.au, zhiyong.chen@newcastle.edu.au}}%
}

\markboth{{\textbf{SUBMITTED} TO IEEE TNNLS : SPECIAL ISSUE-Theory, Algorithms and Applications for Hybrid Intelligent Dynamic Optimization}}%
{Schaich \MakeLowercase{\textit{et al.}}: \thetitle}
%




\maketitle

\begin{abstract}                          

Trajectory optimization  is a fundamental stochastic optimal control problem. This paper deals with a  trajectory optimization approach for {dynamical systems subject to measurement noise that can be {fitted} into linear time-varying stochastic models. Exact/complete solutions to these kind of control problems have been deemed analytically intractable in literature because they come under the category of {Partially Observable {Markov} Decision Processes (POMDPs)}. Therefore, effective solutions with reasonable approximations are widely sought for.} {We propose a reformulation of stochastic control} 
in a reinforcement learning setting. This type of formulation assimilates the benefits of conventional optimal control procedure, with the advantages of maximum likelihood approaches. Finally, an iterative trajectory optimization paradigm {called as Stochastic Optimal} Control - Expectation Maximization (SOC-EM) is put-forth. This trajectory optimization procedure exhibits better performance in terms of reduction of cumulative cost-to-go which is proved both theoretically and empirically. Furthermore, we also provide novel theoretical work which is related to uniqueness of control parameter estimates. {Analysis of the control covariance matrix is presented, which handles stochasticity through
	efficiently balancing exploration and exploitation.}

\end{abstract}


%
\IEEEpeerreviewmaketitle

\section{Introduction}
 
 In recent years, there has been a surge in the research activities related to inference and control of dynamical systems 
in not only the systems and control but also the artificial intelligence communities. The dynamical systems that make intelligent and optimal decisions 
under uncertainty have been formulated in the category of Markov decision process (MDP) \cite{rawlik2013stochastic}. 
{The trajectory optimization problem aims at designing control policies 
to generate trajectories for an MDP that minimizes some measure of performance. 
More applications have been seen in a wide variety of industrial  processes and robotics with the development of computers.} Researchers utilized stochastic optimal control (SOC) methodologies (see e.g., \cite{stengel1994optimal,kappen2012optimal})
to present a solution to an MDP. A specific type of SOC,  reinforcement learning, has exhibited great performance in handling control related tasks in noisy environment, as well as
generalizing the learnt policies to new behaviors through
experience  \cite{levine2016end,montgomery2016guided}. 

{Reinforcement learning is widely used for solving a MDP by optimizing an objective function
through dynamic programming that involves value iteration and policy iteration; see e.g., \cite{bertsekas1995dynamic}, \cite{puterman2014markov}.   
It can be broadly classified  into model-free and model-based categories. 
For instance, in a model-based setup, reward weighted regression was used to learn complex robot-motor motions in  \cite{kober2011reinforcement} and a variant of differential dynamic programming  to learn advanced robotic manipulation policies \cite{levine2016end}. 
	Model-based policy search has been used in trajectory optimization \cite{levine2014motor}, analytical policy gradients \cite{deisenroth2011pilco} and  information-theoretic approaches \cite{deisenroth2013survey}, etc. The typical methods include
iterative linear quadratic Gaussian (iLQG) approach \cite{tassa2012synthesis}, model predictive control (MPC) \cite{zhang2016learning}, path integral linear quadratic regulator (PI-LQR) \cite{chebotar2017combining}, and 
Bregman alternating direction method of multipliers (BADMM)  \cite{wang2014bregman}. These methods are well known in 
searching optimal parameters for a stochastic control policy by utilizing a quadratic cost-to-go  function and a linearized dynamic model in a closed form (after certain approximation).
 	
Especially for the trajectory optimization problem,  sophisticated  model-based techniques are easy to implement 
without suffering from slow convergence. One can refer to \cite{li2004iterative,levine2014learning} for more 
results in this line of research. On the contrary, model-free methods may suffer from slower trajectory optimization  because of reduced sampling efficiency.   Additionally, the research in \cite{levine2016end,chebotar2017combining} exploited adaptability of model-based methods to rapid changes in the environment, which 
is beneficial in dealing with uncertainty. These advantages motivate the research on new model-based 
reinforcement learning policies in this paper.  }

Due to the similarity between policy search and inference problems using a reinforcement learning objective, increased interests have been seen among statistical researchers who treat optimal control as maximum likelihood   inference. 
{A powerful tool known as expectation maximization (EM) has been  widely utilized for solving maximum likelihood problems
in two steps, i.e., guess of the missing data called latent variables and estimation of parameters that best describes the guess. 
It is an iterative process with the probability of guess increased in each iteration.}
The maximum likelihood technique has gained wide popularity in a broad variety of fields of applied statistics such as signal processing 
and dairy science  \cite{borran2002based,shumway2000time}. 
The EM approach has been utilized for robust estimation of linear dynamical systems in \cite{gibson2005robust} and identification of nonlinear state space models in \cite{schon2011system}. {However, there are rare EM based results available in the field of control, which brings another motivation of this paper to study an EM algorithm for SOC problems.}

{It is worth mentioning that a few early attempts at leveraging the concepts of maximum likelihood for solving a MDP can be found in model-free reinforcement learning. For example,  the early work in \cite{cooper2013method}   provides} an evidence of utilizing likelihoods and cost for solving inference problems. 
Probabilistic control and decision has been studied in \cite{neumann2011variational,ziebart2010modeling(a),ziebart2010modeling(b)}  to tackle SOC problems using the maximum entropy principle. The EM technique has been used in inference for optimal policies to maximize cumulative sum of cost for model-based and model-free learning in \cite{hoffman2009expectation} and \cite{toussaint2006probabilistic}, respectively.
However, these works did not substantially analyze the theoretical nature of the solutions.
Other related results include the concept of likelihood used for SOC design in a binary reward model-free setting \cite{toussaint2006probabilistic,dayan1997using} and the   exploitation of EM to weight the reward factors 
in robot control trajectories  \cite{vlassis2009learning,kober2009policy}.
{Nevertheless, EM has attracted widespread attention in model-free domain but not specifically in the model-based domain.}

{The aforementioned discussion has opened up curtain for the main technical scope of this paper, that is, 
the development of a novel EM based SOC algorithm in a model-based domain. 
The main feature lies in its powerful capacity of exploitation of searched state space and 
 attenuation of measurement and/or environmental noise. 
The aforementioned approaches, e.g., iLQG, MPC, BADMM etc., can be problematic in the presence of noise, which propagates through the state equations to generate a highly stochastic policy. 
On top of this, these existing approaches carry out exploration which immensely aggravates this issue.
It demands an effective exploitation step.  Some other relevant maximum likelihood strategies, e.g., \cite{hoffman2009expectation,toussaint2009robot}, may shed light on model-based EM optimization but also suffer from similar disadvantage of handling noise.

Measurement noise in an MDP results in a  partially observable Markov decision process (POMDP).
For instance, the approach in \cite{porta2006point} addressed the optimal control problem for POMDPs with a linear-Gaussian transition model and a mixture of Gaussians reward model, but it requires the action space to be discretized. 
The EM based optimal control proposed in \cite{toussaint2006probabilistic} considers estimation of control covariance matrix, but it does not dig deep into the analysis of covariance matrix that quantifies the trade-off between exploration and exploitation of state space in a reinforcement learning environment. 
Furthermore, the technique of belief space planning by \cite{platt2010belief} aims to transform the partially observable problem into a belief space problem and then it provides an optimal belief-LQR deterministic policy by taking a major step towards effectively handling uncertainty in the system. However, belief-LQR does not deal with stochastic policies, which
restricts the exploration mechanism of a reinforcement learning framework. 



Based on the above discussions about the state-of-the-art SOC methodologies for MDPs or POMDPs, 
it is a promising target of this paper to utilize the advantages of model-based trajectory-centric optimization paradigms together with probabilistic inference based techniques, specifically, EM, to establish an optimal policy in the presence of measurement noise.
For this purpose, the main contributions of this paper are summarized as follows.


}

 {\begin{itemize}
\item  A complete architecture of EM based probabilistic inference algorithm is developed for obtaining stochastic optimal policy parameters.

%
 
\item The algorithm shows the benefits of integrating a model-based optimal control procedure with the advantage of maximum likelihood to deliver an iterative trajectory optimization paradigm, called SOC-EM.

\item It is theoretically proved that  update of policy parameters in an EM iteration leads to reduction of cumulative cost-to-go
for an SOC problem, resulting in (approximate) optimal policy parameters. 

\item The uniqueness property of the maximizer of a surrogate likelihood function is theoretically laid out 
which offers a practically feasible lower dimensional approximation for each EM iteration. 

\item It is exhibited that EM-SOC offers efficient exploitation  of the highly uncertain exploration state space, 
which is analytically quantified by the convergence of control policy covariance matrices to 0
and numerically verified by improved state trajectories with reduced stochasticity in control actions.

\item The effectiveness of EM-SOC in handling measurement noise is explicitly demonstrated.  

\end{itemize}
}

%

%

\section{Preliminaries and Problem Formulation} \label{sec:preliminaries}

This section introduces the dynamic model under investigation, 
the problem formulation, and the proposed solvability procedure. 
It also elaborates the mathematical notations involved with addressing the problem that will be put forth in this paper.
Readers can refer to the symbols summarized in Table~\ref{summary}.

 
 \begin{table}[t]
\caption{Summary of symbols}
\centering
\begin{tabular}{ |p{2.5cm}|p{5.5cm} | } 
 \hline
\textbf{Symbol}& \textbf{Definition} \\
 $k$ & Time instant \\
 $T$ & Length of episode    \\
 $\mathbf{s}_k $   &  Measured state  { (implementation) } \\
   & or { latent state (optimization) at time instant $k$ }\\
 $\mathbf{a}_k$   & Control action  at time instant $k $  \\
  $\mathbf{x}_k $   & Real state  at time instant $k $ \\
 $Y_k(\mathbf{s}_k,\mathbf{a}_k)$   or 
  $Y_k(\mathbf{s}_k,\phi_k)$   
 & Instantaneous cost at time instant $k$ \\
$y_k$ & Observed cost  $p(Y_k)$ \\
{${\mathbb{S}_{T+1}}$} &  {Measured or latent variable $\{\mathbf{s}_1,\mathbf{s}_2,\cdots, \mathbf{s}_{T+1}\}$}\\
 ${\mathbb{Y}_{T}}$   & Reward observation   $\{y_1, y_2, y_3,\cdots, y_{T}\}$ \\
$\phi$ & Controller parameter\\
$\hat{\phi}^{i}$ & Estimation of controller parameter $\phi$ at the  $i$-th iteration\\
$\mathbb{E}$ & Expectation of a random variable \\ 
 $V_\phi ( {\mathbb{S}_{T+1}})$ & Cumulative sum of expected costs \\
  $L_\phi ({\mathbb{Y}_T})$  & Observation log-likelihood \\
 $\mathcal{L}(\phi,\hat{\phi}^i)$&  Mixture likelihood \\
$\text{vec}(\cdot )$ & Column vector stacked by columns of its matrix argument \\
$\col (\cdots)$ & Column vector stacked by its vector arguments \\
$\Tr (\cdot)$ & Trace of its matrix argument\\
$\nabla$ &  Gradient vector field of a scalar function \\
$\nabla^2$  & Hessian matrix; second-order partial derivative of a scalar function\\
 $^\top$ & Transpose operator \\
 $\otimes$ & Kronecker product operator \\
 $\mathbb{R}$ /  $\mathbb{R}^+$ & Set of real numbers / positive numbers \\
$ \mathbf{I}_{(s)} $ & Identity matrix (of dimension $s$)  \\
 \hline
\end{tabular}\label{summary}
\end{table}

\subsection{Mathematical notation and modeling}

The paper takes into account a stochastic dynamics that does not have a known model
from first principles,  in the presence of uncertainties such as parameter variation, external disturbance, 
sensor noise, etc.
The completed system is considered to be a global model, $O$, that is composed of multiple local models 
$o^l,\;  l = \{1, 2, \cdots\}$, and each of which follows an MDP, 
called a local model. 
We are   interested in a finite-horizon optimal control for a particular initial state, rather than for all possible initial states.

The POMDP has a {\it latent state} $\mathbf{s}_{k} \in \mathbb{R}^{n_s}$ 
and a {\it control action} $\mathbf{a}_k \in \mathbb{R}^{n_a}$, at  time instant $k = 1, 2, \cdots$, and the 
local state transition dynamic model is represented by a conditional probability density function (p.d.f.), i.e., 
\begin{align} \label{modelpdf}
 p (\mathbf{s}_{k+1} | \mathbf{s}_k, \mathbf{a}_k).
 \end{align}
In particular, for $k=1$,  $\mathbf{s}_1 \in \mathbb{R}^{n_s}$ is called the {\it initial state}, obeying a specified distribution. Variables $n_s$ and $n_a$ are integers which are the dimensions of state
and action space.
 We specifically consider a finite-horizon MDP in this paper  for $k = 1, 2, \cdots, T$, called an 
{\it episode}, with the time instant $T$ being the end of episode. It is  worth mentioning that the 
p.d.f. in \eqref{modelpdf} varies with time $k$ and 
the time-varying nature is capable of characterizing more complicated dynamical behaviors but also brings more challenges 
in control design. It will be elaborated in Section~\ref{learning_section}. 
 
 The entity $Y_k(\mathbf{s}_k,\mathbf{a}_k) \in \mathbb{R}^+$ denotes the instantaneous real valued {\it cost} for executing action $\mathbf{a}_k$ at state $\mathbf{s}_k$. 
It has a more specific expression as follows,
\begin{align} \label{quad_reward}
Y_k(\mathbf{s}_k,\mathbf{a}_k) = (\mathbf{s}_k-\mathbf{s}^*)^\top \mathbf{Q_s} (\mathbf{s}_k-\mathbf{s}^*) + (\mathbf{a}_k-\mathbf{a}^*)^\top \mathbf{Q_a} (\mathbf{a}_k-\mathbf{a}^*),
\end{align} where $\mathbf{s}^*$ and $\mathbf{a}^*$ are the target state and 
control action, respectively, and  $\mathbf{Q_s} > 0$ and $\mathbf{Q_a}> 0$ are some specified matrices. 
As $\mathbf{s}_k$ and $\mathbf{a}_k$ are random variables,  
$ Y_k(\mathbf{s}_k,\mathbf{a}_k)$  (with $Y_k$ a continuous and deterministic function) is also a random variable, shorted 
as $Y_k$.  We develop another variable, i.e., $y_k=p(Y_k) \in \mathbb{R}^+$ (known as \textit{observed cost}) which is the exponential transformation of the immediate  cost $Y_k(\mathbf{s}_k,\mathbf{a}_k)$ following a p.d.f. 
$ p (y_k | \mathbf{s}_k, \mathbf{a}_k)$, which  will be later elaborated.

Overall, the MPD consists of the transition dynamics $p(\mathbf{s}_{k+1}|\mathbf{s}_k,\mathbf{a}_k)$ and the cost observation p.d.f. $p(y_k|\mathbf{s}_k,\mathbf{a}_k)$ in an augmented form, i.e.,{
 \begin{align} \label{ltv_eq}
     { p \Big(  \begin{bmatrix}
    \mathbf{s}_{k+1}       \\
    y_k      
\end{bmatrix} | \mathbf{s}_k, \mathbf{a}_k \Big) }   = \mathcal{N} \Big( { \mathbf{A}^o_k }
{ \begin{bmatrix}
    \mathbf{s}_k       \\
    \mathbf{a}_k       
\end{bmatrix} }  ,{ \pmb{\Sigma}^o_k  \Big)},
\end{align} 
which is referred to as the {\it dynamic model} in the subsequent parts of the paper.}
A time-varying linear Gaussian p.d.f. is used in  \eqref{ltv_eq} as 
an approximation of a real model which is in general nonlinear,  where 
the matrices $ \mathbf{A}^o_k$ and $\pmb{\Sigma}^o_k$ in  to be determined in Section~\ref{learning_section}
using the dynamic model fitting technique.

\begin{figure}[t]
  \centering
\includegraphics[scale=0.52, bb=150 470 360 800]{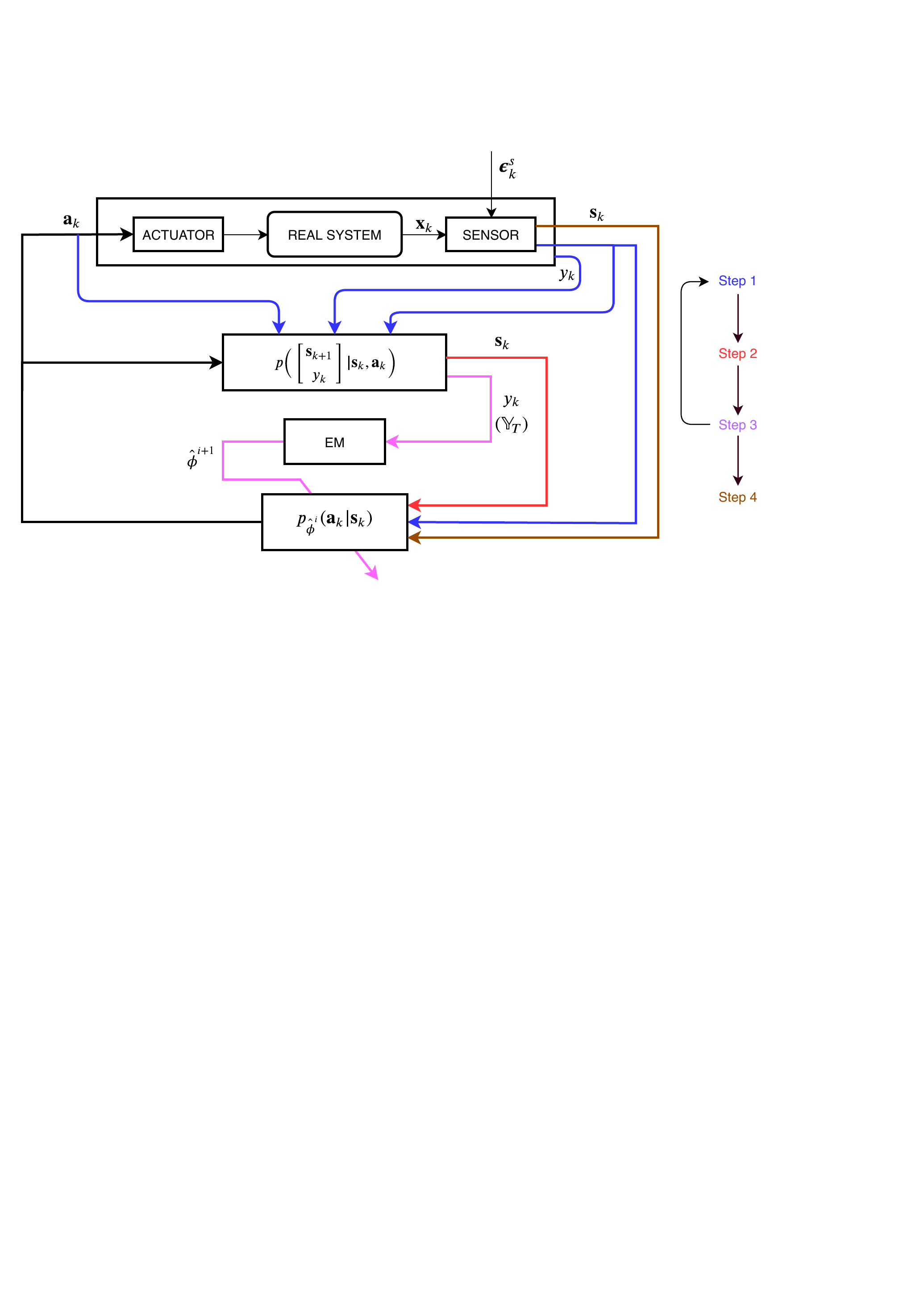}
   \caption{Schematic diagram of the overall design procedure in four steps.}
    \label{fig:overall}
\end{figure}

\subsection{Controller parameter space} 
This subsection presents the definition of parameter space of a controller that is utilized in the paper. The control action is sampled from a  linear Gaussian p.d.f. that describes the policy as shown below,
\begin{align} 
        p_{{\phi_k}}({\mathbf{a}}_k|\mathbf{s}_k) = \mathcal{N} ( \mbF_k \mathbf{s}_k + \mbe_k ,{\mbSigma_k} )  ,  \label{control}
\end{align}
for some matrices $\mbF_k, \mbSigma_k$ and a vector $\mbe_k$, representing state feedback control. {\bluecolor }
The matrix $\mbSigma_k$ is symmetric positive definite, and
$\mbSigma_k^{\frac{1}{2}}$ is the square root of ${\mbSigma_k}$ satisfying 
  ${\mbSigma_k}  =(\mbSigma_k^{\frac{1}{2}}) ^\top \mbSigma_k^{\frac{1}{2}} $. 
Let $\mbf_k =  \text{vec} (\mbF_k)$ and
$\mbsigma_k =  \text{vec} (\mbSigma_k^{\frac{1}{2}} )$.
Then,  the vector 
\begin{align*} 
{\phi}_k =\col (\mbf_k, \mbe_k, \mbsigma_k),
\end{align*} is called the  controller parameter vector. 
Over the episode under consideration, the controller parameters are lumped as follows, 
\begin{align}
\phi  =\col (\phi_1, \phi_2, \cdots, \phi_T) \in \Phi,
\end{align}
{where ${{\Phi}}$ is a non-empty convex compact subset of $\mathbb{R}^{({n_a n_sT+n_aT} + {n_a n_aT} ) }$.
The time-varying feature of the controller is represented by the variation of $\phi_k$ with $k$,
 which aims to account for the complexity of the dynamical system.
A stochastic policy adopted in this paper lays down the exploration mechanism  in a reinforcement learning setting. 
 }

\subsection{Problem formulation} \label{section:DP}

{This section starts with a definition of the conventional SOC problem (see e.g., \cite{stengel1994optimal})
and then moves to the specific formulation of the problem studied in this paper.}
For the stochastic dynamic model \eqref{ltv_eq},  the conventional SOC problem is formulated as follows,  
\begin{align} \label{cumsum_obj}
 & \min_{\mathbf{a}_1, \mathbf{a}_{2}, . . ,\mathbf{a}_T} \mathbb{E} \sum_{k=1}^T Y_k(\mathbf{s}_k,\mathbf{a}_k),
\end{align}
where $\mathbf{s}_k$ and $\mathbf{a}_k$ are the
variables of the dynamic model p.d.f.  \eqref{ltv_eq}
and  the control action p.d.f. \eqref{control}.
The expectation is taken over the measurement states which are a result of instantiations 
of the noise in the dynamical equation with an initial state $\mbs_1$. To express the cost penalty to be explicitly dependent on 
$\phi_k$, we rewrite 
$Y_k(\mathbf{s}_k,\mathbf{a}_k)$ as $Y_k(\mathbf{s}_k,\phi_k)$ with slight abuse of notation. 
 Also, we can rewrite \eqref{cumsum_obj} in terms of the controller parameter $\phi$, i.e., 
  \begin{align} \label{cumsum_obj1}
 &  \min_{\phi}  \mathbb{E}  V_\phi ( {\mathbb{S}_{T+1}})  \;
 \text{for} \;
 V_\phi ( {\mathbb{S}_{T+1}}) \triangleq   \sum_{k=1}^T Y_k(\mathbf{s}_k,\phi_k).
\end{align}
with ${\mathbb{S}_{T+1}}  =   \{\mathbf{s}_1,\mathbf{s}_2,\cdots, \mathbf{s}_{T+1}\}$.

{A complete solution to the optimization problem \eqref{cumsum_obj} is hardly analytically tractable
 \cite{platt2010belief,10.2307/3689975}. 
 It is more realistic to pursue effective solutions with reasonable approximations
 as seen in numerous references including \cite{levine2014motor,tassa2012synthesis,li2004iterative,ziebart2010modeling(a)}.  
  It is worth mentioning that,
 existence of a stationary policy for a POMDP is NP-complete \cite{littman1994memoryless,lusena2001nonapproximability}. Even a solution to a finite horizon POMDP is shown to be PSPACE-complete for discrete states, actions and observations \cite{10.2307/3689975}. Therefore, 
in this paper we propose a new problem formulation with a procedure of solution 
that can be regarded as a decent alternative to the problem \eqref{cumsum_obj}. 
The procedure is elaborated below in a four step architecture and also illustrated in Fig.~\ref{fig:overall}. 
  }

{\it Step 1: Dynamic model fitting:}
From an initial state $\mathbf{s}_1$ sampled from a specified distribution, 
 the real system is operated with the controller \eqref{control}
for a pre-selected controller parameter $\phi = \hat\phi^0=\col (\hat\phi^0_1,\hat \phi^0_2, \cdots, \hat\phi^0_T)$
and the control actions $\{\mathbf{a}_1,\mathbf{a}_2,\cdots, \mathbf{a}_{T}\}$ 
and the states
$\{\mathbf{s}_1,\mathbf{s}_2,\cdots, \mathbf{s}_{T+1}\}$ are recorded.
Calculate $Y_k(\mathbf{s}_k,\hat\phi^0_k)$ and hence $y_k = e^{ - Y_k }$. 
{Then,  the dynamic model \eqref{ltv_eq} is identified by fitting it to 
to the collected tuples of data $\{\mathbf{s}_k,\;\mathbf{a}_k,\;\mathbf{s}_{k+1},\;y_k\}$, $k=1,\cdots, T$.}

{\it Step 2: Generation of cost observation:}
From an initial state $\mathbf{s}_1$ sampled from a specified distribution, 
the   cost observations ${\mathbb{Y}_{T}}= \{y_1, y_2, y_3,\cdots, y_{T}\}$ are generated using
the dynamic model \eqref{ltv_eq}  (obtained from Step 1)
and the controller  \eqref{control} with the controller parameter $\phi = \hat\phi^0$.  
 
{\it Step 3: Optimization of control action:}  Let  ${ p_{\phi} (\mathbb{S}_{T+1} | \mathbb{Y}_T)}$ be the probability of 
the {\it latent states}
${\mathbb{S}_{T+1}}  = \{\mathbf{s}_1,\mathbf{s}_2,\cdots, \mathbf{s}_{T+1}\}$
given the observation ${\mathbb{Y}_{T}}$ (obtained from Step 2), obeying the 
closed-loop system composed of the dynamic model \eqref{ltv_eq}  (obtained from Step 1)
and the controller  \eqref{control} with a 
controller parameter $\phi$. 
The optimization of a local control policy is formulated as follows
\begin{align} \label{valuefn_ltv}
	\phi^{*} = \arg\min_{\phi}\mathbb{E}_{ p_{\hat\phi^0} (\mathbb{S}_{T+1} | \mathbb{Y}_T)} V_\phi ( {\mathbb{S}_{T+1}}).
\end{align}

{\it Step 4: Implementation and evaluation:} Run the real system with the  controller \eqref{control} for the optimal parameter 
$\phi ={ \phi^*}$ and evaluate the performance.

A practical approach to solve the optimization problem \eqref{valuefn_ltv} is to use the following strategy,
 \begin{align} \label{valuefn_ltvi}
{ \hat \phi^{i*}}  = \arg\min_{\phi}\mathbb{E}_{ p_{\hat\phi^i} (\mathbb{S}_{T+1} | \mathbb{Y}_T)} V_\phi ( {\mathbb{S}_{T+1}}) ,
\end{align}
recursively with $\hat \phi^{i+1} = { \hat \phi^{i*}}$,  for $i=0,1, \cdots$.
It is expected that $ \hat \phi^{i} $ approaches $\phi^*$ as $i$ goes to $\infty$.

Throughout the paper, we use the simplified notation
\begin{align}
   \mathbb{E}_{\phi} ( * | \mathbb{Y}_T ) 
   \triangleq   {\mathbb{E}_{p_{\phi} (\mathbb{S}_{T+1}|\mathbb{Y}_T)} (  *) } \label{Ltheta_kheta1}
\end{align}
and  \eqref{valuefn_ltvi} can rewritten as 
 \begin{align} \label{valuefn_ltv2}
{ \hat \phi^{i*}}  = \arg\min_{\phi}\mathbb{E}_{ {\hat\phi^i}} ( V_\phi ( {\mathbb{S}_{T+1}}) | \mathbb{Y}_T ).
\end{align}
After each iteration $i$, one has an updated controller parameter $ \hat \phi^{i+1}$ and Steps 1 and 2
are repeated with $\phi =\hat \phi^{i+1}$ for an updated dynamic model and updated  cost observation.

 {\begin{remark} 
In Steps 1 and 4, the real dynamical system is operated for data generation and performance evaluation, respectively. 
The state $\mathbf{s}_k$ is physically measured, which represents the observed system state carrying measurement noise. 
 However, in  Steps 2 and 3, only theoretical computation is conducted without operating the real system, thus leveraging the latency nature of states.
  Here,  $\mathbf{s}_k$ represents the explored state
obeying a  joint probability  ${ p_{\phi} (\mathbb{S}_{T+1} | \mathbb{Y}_T)}$ 
 conditioned on the observation ${\mathbb{Y}_{T}}$ and parameterized with a given $\phi =\hat \phi^i$ at each iteration.
 It is thus called a latent state. In both cases, either the states  carrying noise or the 
 explored state samples are adopted. In other words, the real/true system states are not observed, 
 with which the model \eqref{ltv_eq} is treated as a POMDP. 
 \end{remark}

 \begin{remark} 
The subsequent sections are concerned about the optimization problem \eqref{valuefn_ltv2}
which is regarded as the approximation of the original optimization problem \eqref{cumsum_obj}. 
It is easy to see that \eqref{valuefn_ltv2}
is equivalent to   \begin{align}  \label{valuefn_ltv3}
 \min_{\mathbf{a}_1, \mathbf{a}_{2}, . . ,\mathbf{a}_T} \mathbb{E}_{p_{\hat\phi^i} (\mathbb{S}_{T+1}|\mathbb{Y}_T)}
   \sum_{k=1}^T Y_k(\mathbf{s}_k, \mathbf{a}_k)   .
\end{align}
It is expected that recursively solving the problem  \eqref{valuefn_ltv2} or \eqref{valuefn_ltv3} 
will approach a solution to \eqref{cumsum_obj}.
However, the global convergence of the recursion is of great challenge and the effectiveness can only be 
numerically verified in this paper. 
The gap between \eqref{cumsum_obj} and   \eqref{valuefn_ltv3} is further discussed as follows. 
The  optimization in  \eqref{valuefn_ltv3} can be intuitively interpreted as finding a probability distribution of state trajectories whose  samples contain lowest expected cost-to-go. One can refer to Section-3.3 of \cite{levine2014motor} which describes a similar type of objective function for achieving trajectories of lowest cost. 
 It can be a reasonably approximation of the real time cost-to-go in  \eqref{cumsum_obj}.
 \end{remark}

%
 

}

In the remaining sections,  we first elaborate Step 1, the dynamic model fitting procedure,
in Section~\ref{learning_section}.
The main technical challenges in optimization of the control action \eqref{valuefn_ltv2}, accounting for Steps~2 and 3, 
are addressed in Sections~\ref{sec:trajectory_optimization} and
 \ref{section:EMsolution}, in a novel systematic framework. 
Step~4 is discussed in Section~\ref{sec:results}.

\section{Dynamic Model Fitting} \label{learning_section}

%
%

In this section, we elaborate the procedure formulated in Step 1 to attain linear time-varying parameter estimates of the dynamic model
\eqref{ltv_eq}. Technically, we merge the procedure adopted in \cite{levine2016end} with the existing variational Bayesian (VB) strategies for a finite mixture model that can be referred to in \cite{bishop2006pattern}.

We first give a specific definition of $y_k$ as follows,
 \begin{align} \label{expo_transformation}
   y_k(\mbs_k,\mba_k)  = e^{-  Y_k (\mbs_k,\mba_k)}.
 \end{align}
{Intuitively,   $y_k(\mbs_k,\mba_k) \in (0,1]$ characterizes the likelihood of $(\mbs_k,\mba_k)$ being near the optimal trajectory. 
When an action results in a less cost $Y_k (\mbs_k,\mba_k)$, it implies a larger $y_k(\mbs_k,\mba_k)$ representing 
a higher likelihood of being near the optimal trajectory.
Such an exponential transformation has been proved successful in 
determining the probability of occurrence of an optimal event in optimal control; see, e.g., \cite{cooper2013method,toussaint2009robot}.
}

%
%

 As described in the aforementioned Step 1, we can run one experiment and collect the tuples 
$\{\mathbf{s}_k,\;\mathbf{a}_k,\;\mathbf{s}_{k+1},\;y_k\}$ for every episode $k=1,\cdots, T$. In practice, 
the experiments can be repeated for $M$ times from the same initial conditions with a random seed value to gather sufficiently many samples, each of which is denoted by,
\begin{align*}
\mathcal{D}_k^m=\{\mathbf{s}_k,\;\mathbf{a}_k,\;\mathbf{s}_{k+1},\;y_k\}_{\text{$m$-th experiment}},
\end{align*}
 for $m=1,\cdots,M$.  Let 
$\mathcal{D}_k = \{\mathcal{D}_k^1,\cdots, \mathcal{D}_k^M \}$ 
and $\mathcal{D} = \{\mathcal{D}_1,\cdots, \mathcal{D}_T \}$.
 
  {Research in \cite{levine2016end} suggests that 
utilizing simple linear regression to fit
the data set $\mathcal{D}$ requires a large amount of samples
and may become problematic in high dimensional scenarios. 
However, the linear Gaussian fitting approach has been proven to be effective in reducing sample complexity 
noting that the samples from a dynamical system in adjacent time steps are correlated. 
More specifically,  it is assumed that the data set $\mathcal{D}$ is generated from a mixture of a finite number of Gaussian distributions with unknown parameters, to which one one can fit a Gaussian mixture model (GMM). 
The procedure involves constructing  normal-inverse Wishart distributions to act as prior for means and covariances of Gaussian distributions involved in mixture model. In addition to it, Dirichlet distributions are defined to be the prior on the weights of the Gaussian distributions which would explain the mixing proportions of Gaussians. 
Then, the iterative VB strategy is adopted to increase the likelihood of a joint variational distribution (see e.g., \cite{bishop2006pattern}-Section 10.2)} to determine the parameters of the GMM, i.e., the means, covariances and weights of the Gaussians for a particular time instant $k$.
More specifically, the Gaussian distribution is of the form
\begin{align} \label{NIWp}
p (\mathbf{s}_k,\mathbf{a}_k,\mathbf{s}_{k+1},y_k) = \mathcal{N} (\mbmu_k, \mbLambda_k),
\end{align}
for the mean $\mbmu_k$ and the covariance ${\mbLambda}_k$.  The parameters $\mbmu_k$ and ${\mbLambda}_k$ are the a-posteriori estimates which are evaluated by a Bayesian update rule with the information of the dataset $\mathcal{D}$ and normal-inverse Wishart prior.

The Gaussian distribution  \eqref{NIWp} can then be conditioned on states and action, i.e., $(\mathbf{s}_k,\mathbf{a}_k)$, using standard identities of multivariate Gaussians, which results in \eqref{ltv_eq}
for the following parameters
 
 \begin{align*}
 \mathbf{A}^{o}_k = \begin{bmatrix}
    \mathbf{A}^d_k      & {\mathbf{B}^d_k} \\
    {\mathbf{A}^r_k}       & {\mathbf{B}^r_k} 
\end{bmatrix} ,  \pmb{\Sigma}^{o}_k = {\begin{bmatrix}
    {\pmb{\Sigma}^d_k} & {{\pmb{\Sigma}^{rd}_k}} \\
    {\pmb{\Sigma}^{rd}_k}^\top      & { {\pmb{\Sigma}^r_k}}
\end{bmatrix}  } .
\end{align*}
The dimensions of the matrices are  $\mathbf{A}^d_k \in \mathbb{R}^{n_s \times n_s}$, ${\mathbf{B}^d_k} \in \mathbb{R}^{n_s \times n_a}$, ${\pmb{\Sigma}^d_k} \in \mathbb{R}^{n_s \times n_s}$, ${{\mathbf{A}^r_k}} \in \mathbb{R}^{1 \times n_s}$, ${\mathbf{B}^r_k} \in \mathbb{R}^{1 \times n_a}$, ${ {\pmb{\Sigma}^r_k}} \in \mathbb{R}$, $ \mathbf{A}^o_k \in \mathbb{R}^{({n_s + 1}) \times (n_a +n_s) }$ and ${ {\pmb{\Sigma}^o_k} }  \in \mathbb{R}^{(n_s+1) \times (n_s+1)}$.

In the dynamic model \eqref{ltv_eq}, the term $ {\pmb{\Sigma}^{rd}_k}$ denotes the correlation between $\mathbf{s}_{k+1}$ and $y_k$. Without loss of generality, we assume that  ${\pmb{\Sigma}^{rd}_k}=0$. Note that one can also consider ${\pmb{\Sigma}^{rd}_k} \neq 0$ and utilize methods of de-correlation to carry out the entire procedure in a similar way. It is assumed that the covariance matrices are symmetric positive definite, that is, $ {\pmb{\Sigma}^d_k} > 0$,  ${ {\pmb{\Sigma}^r_k}}> 0 $, and  ${\pmb{\Sigma}^o_k}> 0$, throughout the paper.

We consider the dynamic model \eqref{ltv_eq} for the episode $k=1,\cdots,T$, assuming the 
initial time $k=1$.  This kind of modeling resembles with   pre-existing studies in, e.g., \cite{levine2016end,montgomery2016guided,tassa2012synthesis,zhang2016learning,li2004iterative,levine2014learning}.  It is noted that shifting the model \eqref{ltv_eq} by $k_0 \geq 0$ gives a model as follows,  in the new episode $k= k_0+1,\cdots,k_0+T$,  \begin{align}
     { p \Big(  \begin{bmatrix}
    \mathbf{s}_{k+1}       \\
    y_k      
\end{bmatrix} | \mathbf{s}_k, \mathbf{a}_k \Big) }  = \mathcal{N} \Big( { \mbA^o_{k-k_0} }
{ \begin{bmatrix}
    \mathbf{s}_k       \\
    {a}_k       
\end{bmatrix} }  ,{  \pmb{\Sigma}^o_{k-k_0}  \Big)}. \label{shiftmodel}
\end{align}
Therefore, the time-varying feature of the linear Gaussian model \eqref{ltv_eq} is not absolute but relative. By relatively time-varying we mean that the dynamical system parameters $ \mbA^o_{k-k_0}$  and  $\mbSigma^o_{k-k_0}$  in \eqref{shiftmodel} 
do not depend on the absolute time $k$, but on the relative time interval $k-k_0$. In other words, 
the model is independent of the initial time $k_0$. 
The time-varying nature of the model \eqref{ltv_eq} is capable of characterizing the complicated dynamical behaviors studied in this paper by more accurately capturing the nonlinearity in a piecewise-linear Gaussian manner. On the contrary, a time-invariant model with a unique Gaussian distribution in \eqref{ltv_eq} for all $k$ could be oversimplified, inaccurate and would definitely not describe a complicated model. Nevertheless, it is possible to fit only a relatively time-varying model to the collected data   by running multiple experiments  at different time instants. 
  
 
\section{Optimization of Control Action via EM} \label{sec:trajectory_optimization}

This section starts with some concepts used in the well acknowledged EM algorithm.  Basically, EM computes the maximum likelihood estimate of some parameter vector $\phi$ (whose design is at the discretion of the user), say $\hat{\phi}_{EM}$ based on an observed data set $\mathbb{Y}_T$. In particular, the likelihood of observing the data $\mathbb{Y}_{T}$ written as $p_\phi(\mathbb{Y}_T)$ 
does not decrease in an iterative manner, i.e.,
\begin{equation} \label{theta_ml}
   \hat{\phi}_{EM} \in \{  \phi \in \Phi : p_{\phi} (\mathbb{Y}_{T}) \geq p_{\hat{\phi^i}} (\mathbb{Y}_{T}) \},
\end{equation}
where   
$\hat{\phi^i}$ is a (known) considerably good parameter estimate with which the EM approach is initialized (at the
iteration labeled $i$). 


The  EM algorithm involves 
the \textit{observation log-likelihood,}
 \begin{align}
  L_\phi ({\mathbb{Y}_T}) & {\triangleq \log p_\phi(  {\mathbb{Y}_{T}} )} \label{Lphi} 
\end{align}
and an essential approximation of  log of \textit{mixture likelihood} of some latent variables ($\mathbb{S}_{T+1}$) and the observations ($\mathbb{Y}_{T}$) with a surrogate function $ \mathcal{L}(\phi,\hat{\phi}^{i})$ defined  in the following equation, 
\begin{align}
 \mathcal{L}(\phi,\hat{\phi}^{i})  &  \triangleq \mathbb{E}_{\hat{\phi}^i}   (  \log  p_\phi (\mathbb{S}_{T+1}, \mathbb{Y}_T) |\mathbb{Y}_T)  
    .  \label{Ltheta_kheta}
\end{align}
 {It is assumed that both $L_\phi ({\mathbb{Y}_T}) $ and $\mathcal{L}(\phi,\hat{\phi}^{i})$ are 
 differentiable in $\phi\in\Phi$.}
Some lemmas used for the EM algorithm are given in Appendix.

 Next, we aim to propose an  EM based method for solving the optimal control problem 
\eqref{valuefn_ltv2} associated with the dynamic model \eqref{ltv_eq} and the controller  \eqref{control}, as formulated in the aforementioned Step~3. 
To bridge the relationship between the 
optimal control problem and the EM algorithm that is originally used for maximizing 
the likelihood of observed data, we first recall the observation $\mathbb{Y}_{T}$ in Step~2.
Let $\mathbb{S}_{T+1}$ be the latent states whose probability is denoted as 
 ${ p_{\phi} (\mathbb{S}_{T+1} | \mathbb{Y}_T)}$, given the observation ${\mathbb{Y}_{T}}$, obeying the 
closed-loop system composed of the dynamic model \eqref{ltv_eq}
and the  policy  \eqref{control} with a parameter $\phi$. 
More specifically, one has
 \begin{align} 
 \label{pphi_sys}
       p_\phi ({\mathbb{S}_{T+1}}, {\mathbb{Y}_{T}})    = p(\mathbf{s}_1) \prod_{k=1}^{T} p_{\phi_k} (\mathbf{s}_{k+1},y_k|\mathbf{s}_k).
 \end{align}
Hence,  we can define $L_\phi ({\mathbb{Y}_T})$  and $\mathcal{L}(\phi,\hat{\phi}^{i})$ as
in \eqref{Lphi} and \eqref{Ltheta_kheta}.

In the conventional EM, it has been revealed (see Lemma~\ref{lemma:EM_proof})
that, in a recursive procedure, 
a new parameter $\phi=\hat{\phi}^{i+1}$ that increases $\mathcal{L}(\phi, \hat{\phi}^i)$
from $\phi=\hat{\phi}^{i}$, also increases $L_\phi(\mathbb{Y}_T)$.
We aim to further prove that, the new parameter $\phi=\hat{\phi}^{i+1}$ also decreases 
$\mathbb{E}_{ {\hat\phi^i}} ( V_\phi ( {\mathbb{S}_{T+1}}) | \mathbb{Y}_T )$ in \eqref{valuefn_ltv2}, thus bridging 
the EM algorithm 
and the optimal control objective. It can be simply stated that the  EM algorithm for finding
$\hat\phi^{i*}$  in \eqref{maxlmaxcalL} with $\hat\phi^{i+1}  = \hat\phi^{i*}$ also
works for   \eqref{valuefn_ltv2}.


{The theorem  to be established in this section is based on the following assumption for the distribution of $Y_k$.

 \begin{assumption} \label{assumption:pdf_Yk}
The p.d.f. of $Y_k$  follows an exponential distribution with parameter $\lambda$, i.e.,
\begin{align}\label{pYk}
 p(Y_k) = \lambda e^{-\lambda Y_k} \text{ where }  \lambda>1.
\end{align}

\begin{remark} The above assumption is practically reasonable for the following two reasons.
First,  both $\mbs_k$ and $\mba_k$ follow a Gaussian distribution in Section~\ref{learning_section}, 
therefore $Y_k$ follows a linear combination of independent non-central chi-squared variables with some degrees of freedom. Solving for a p.d.f. of $Y_k$ is complicated (see e.g. Appendix~A.1 of \cite{paolella2018linear}). As all these distributions are related to a general exponential family, it is reasonable to assume that  $Y_k$ also follows an exponential distribution. 
Second, the justification for using an exponential distribution can also be found in 
relevant work. For example, it is assumed that rewards (negative costs) are drawn from an exponential distribution
in \cite{dayan1997using} and  a so-called exponentiated payoff distribution is used in \cite{norouzi2016reward} as a link between maximum likelihood and an optimal control objective. 
\end{remark}
 \end{assumption}
}
Then,  we can give the following lemma regarding the distribution property of $y_k$
defined in \eqref{expo_transformation}, which is of sole importance for establishing a theoretical relationship between 
the mixture likelihood function and the SOC objective. 


\begin{lemma} \label{lemma_prob2}
For $Y_k$ of  the p.d.f. \eqref{pYk},
the random variable $y_k$ in \eqref{expo_transformation} has a p.d.f. of the form 
\begin{align} \label{pdf_y_k}
   p(y_k) = {\lambda} { y_k  }^{{\lambda} -1} . 
\end{align}
\end{lemma}

\begin{proof}
 The random variable of $y_k$ has the following cumulative distribution function
\begin{align*}
    \mathbb{F}_{y_k} (x) &= p( y_k < x) = p (e^{- Y_k} <x) = p(  Y_k > {-} \log  (x) ).
    \end{align*}
Further calculation implies
\begin{align*}
   \mathbb{F} (x) & = \int_{{-} \log  (x)  }^\infty p(Y_k) dY_k=
   \int_{{-} \log  (x)  }^\infty  \lambda e^{-\lambda Y_k}  dY_k \\& =
-  e^{-\lambda  \infty} + e^{ \lambda \log x} 
    =  {x  }^{ {\lambda}  } .
\end{align*}
Thus differentiating $\mathbb{F} (x)$ with respect to $x$ gives the  p.d.f of $y_k$ as
$p(x) = d \mathbb{F} (x)  /dx =  {\lambda} { x }^{{\lambda} -1}$,  which is simply denoted as \eqref{pdf_y_k}.
\end{proof}

Now,  the main result is stated in the following theorem. 
Recall that $Y_k$ can be explicitly expressed by $Y_k(\mathbf{s}_k,\phi_k)$, 
and accordingly, $y_k$ by $y_k(\mathbf{s}_k,\phi_k)$, which is used in the proof of the theorem.

\begin{theorem} \label{thm:main_theorem}
Suppose the parameter $\hat{\phi}^{i+1}$ is produced such that
\begin{align}  \label{calLi}
 \mathcal{L}( \hat{\phi}^{i+1}, \hat{\phi^i})  \geq    \mathcal{L} ( \hat{\phi^i} , \hat{\phi^i}).  \end{align}
Then,  the cumulative sum of expected costs defined in \eqref{valuefn_ltv2} satisfies
\begin{align} \label{EVi}
  \mathbb{E}_{ {\hat\phi^i}} ( V_{\hat\phi^{i+1}} ( {\mathbb{S}_{T+1}}) | \mathbb{Y}_T )  \leq 
\mathbb{E}_{ {\hat\phi^i}} ( V_{\hat\phi^i} ( {\mathbb{S}_{T+1}}) | \mathbb{Y}_T ).
\end{align}
 \end{theorem}
 
\begin{proof}
First, by Lemma~\ref{lemma:EM_proof}, \eqref{calLi} implies 
\begin{align}
L_{\hat{\phi}^{i+1}} (\mathbb{Y}_T)  - L_{\hat{\phi}^i} (\mathbb{Y}_T) \geq 0.
\end{align}
Denote $\hat{\phi}^{\imath}=[(\hat{\phi}^{\imath}_1)^\top,\cdots, (\hat{\phi}^{\imath}_T)^\top]^{\top}$ for $\imath =i, i+1$.
One has
\begin{align*}
 L_{\hat{\phi}^\imath} (\mathbb{Y}_T)
  = &  \log  p_{\hat{\phi}^\imath} (\mathbb{Y}_T)  \\
 = &  \mathbb{E}_{ p_{\hat{\phi}^{i}} (\mathbb{S}_{T+1} | \mathbb{Y}_T)}   [\log  p_{\hat{\phi}^\imath} (\mathbb{Y}_T) ]  \\
      =&   \mathbb{E}_{ p_{\hat{\phi}^{i}} (\mathbb{S}_{T+1} | \mathbb{Y}_T)}   [ \sum_{k=1}^{T} \log  p_{\hat{\phi}_k^\imath} ( y_k ) ].\end{align*}
With $p_{\hat{\phi}_k^\imath} ( y_k )  =  p  (y_k(\mathbf{s}_k,\hat\phi^\imath_k) )$, the above calculation continues as follows, 
by utilizing the results of Lemma \ref{lemma_prob2},
\begin{align*}
 L_{\hat{\phi^\imath}} (\mathbb{Y}_T)
  = & \mathbb{E}_{ p_{\hat{\phi}^{i}} (\mathbb{S}_{T+1} | \mathbb{Y}_T)}  [ \sum_{k=1}^{T} \log  p  (y_k(\mathbf{s}_k,\hat\phi^\imath_k) )] \\
 =&  \mathbb{E}_{ p_{\hat{\phi}^{i}} (\mathbb{S}_{T+1} | \mathbb{Y}_T)}  [\sum_{k=1}^T \log \lambda (  y_k(\mathbf{s}_k,\hat\phi^\imath_k) )^{\lambda-1}  ] \\
 =& \mathbb{E}_{ p_{\hat{\phi}^{i}} (\mathbb{S}_{T+1} | \mathbb{Y}_T)} [\sum_{k=1}^T (\lambda-1) (-
Y_k( \mathbf{s}_k,\hat\phi^\imath_k )) ] +T\log \lambda\\
 =& - (\lambda-1)  \mathbb{E}_{ p_{\hat{\phi}^{i}} (\mathbb{S}_{T+1} | \mathbb{Y}_T)} [\sum_{k=1}^T Y_k( \mathbf{s}_k,\hat\phi^\imath_k ) ] +T  \log  \lambda\\
 =& - (\lambda-1) \mathbb{E}_{ p_{\hat{\phi}^{i}} (\mathbb{S}_{T+1} | \mathbb{Y}_T)} [\sum_{k=1}^T Y_k( \mathbf{s}_k,\hat\phi^\imath_k ) ] +T  \log  \lambda.
\end{align*}
Next, from \eqref{cumsum_obj1}, i.e., $V_{\hat\phi^\imath} ({\mathbb{S}_{T+1}})  = \sum_{k=1}^TY_k( \mathbf{s}_k,\hat\phi^\imath_k )$, one has
\begin{align}\label{eq:maineq_theoremv1}
 L_{\hat{\phi^\imath}} (\mathbb{Y}_T)
 =& - (\lambda-1)  \mathbb{E}_{\hat{\phi}^{i}}  (V_{\hat\phi^\imath} ({\mathbb{S}_{T+1}})   | \mathbb{Y}_T )  +T  \log  \lambda.
\end{align}
As a result, 
\begin{align*} 
0\leq  &L_{\hat{\phi}^{i+1}} (\mathbb{Y}_T) -L_{\hat{\phi^i}} (\mathbb{Y}_T) \\
 =& - (\lambda-1)    \big[ \mathbb{E}_{\hat{\phi}^{i} } (V_{\hat\phi^{i+1}} ({\mathbb{S}_{T+1}})  | \mathbb{Y}_T)   - 
 \mathbb{E}_{\hat{\phi}^{i} } (V_{\hat\phi^{i}} ({\mathbb{S}_{T+1}})  | \mathbb{Y}_T)    \big].
 \end{align*}
 It implies \eqref{EVi} and completes the proof. 
\end{proof}

{\begin{remark}
Theorem~\ref{thm:main_theorem} takes into account the exponential transformation according to \eqref{expo_transformation} and
\eqref{pYk} to ensure the decrease of the expected cost-to-go with increased likelihood. It suggests an effective approximate approach for the minimization objective in optimal control through pursuing the maximum likelihood objective. 
This approach is intuitively consistent with some results in literature. For example,  the research in \cite{toussaint2009robot} claimed that maximum likelihood based inference is an approximation of the iLQG-based solution to a SOC problem. 
As an approximate class of inference based techniques, a maximum likelihood method was also used in \cite{kappen2012optimal}. 
Some similar approximate relationship between a reward proportional likelihood objective 
and a policy gradient objective function was revealed in inference based policy search \cite{toussaint2006probabilistic}.

\end{remark}
}

 \section{A Practical Solution to SOC-EM}  \label{section:EMsolution}
 
After having established the relationship between EM and optimal control this paper proceeds towards a closed form solution of 
 $\arg\max_{\phi} \mathcal{L} (\phi, \hat{\phi}^i)$. The first step is   to deliver an explicit expression of the mixture likelihood
 associated with the dynamic model  \eqref{ltv_eq} and the controller \eqref{control}.

\subsection{Explicit expression of mixture likelihood}

The explicit expression of the mixture likelihood $\mathcal{L}(\phi, \hat{\phi^i})$ defined in \eqref{Ltheta_kheta1}  
is given in the following lemma.

\begin{lemma} \label{lemma:LTV}
The function  $\mathcal{L}(\phi, \hat{\phi^i})$ for the dynamic model  \eqref{ltv_eq} and the controller \eqref{control} can be expressed 
as follows,   \begin{align} \label{calLbarcalL}
\mathcal{L} (\phi,\hat{\phi^i})    =\log p(\mathbf{s}_1) +  \sum_{k=1}^T  \bar{\mathcal{L}}_k (\phi_k, \hat{\phi}^i),
\end{align}
for
\begin{align} \label{ltv_surrogate}
\log p(\mathbf{s}_1) =& -\frac{1}{2} \log |\mathbf{P}_1| + (\mathbf{s}_1 - \pmb{\mu}_1)^\top \mathbf{P}_1^{-1} (\mathbf{s}_1 - \pmb{\mu}_1) , \nonumber \\
  {\bar{\mathcal{L}} _k(\phi_k,\hat{\phi^i})} =&   -\frac{1}{2} \Tr \{ {{\pmb{\Sigma}}_k^o}^{-1}(\Theta_1 (\phi_k) - \Theta_2 (\phi_k) {\mathbf{A}^o_k}^{\top}  \nonumber\\
     &   - {\mathbf{A}^o_k}  {\Theta_2(\phi_k)}^{\top} +  {\mathbf{A}^o_k} \Theta_3 (\phi_k) {\mathbf{A}^o_k}^{\top}) \} -\frac{1}{2} \log  |{{{\pmb{\Sigma}}_k^o}}| ,
\end{align}
where
$\pmb{\mu}_1$ and $\mathbf{P}_1$ are some known mean and covariance of the initial state $\mathbf{s}_1$
and the other terms are defined by 
\begin{align}
\Theta_1 (\phi_k) & = \mathbb{E}_{\hat\phi^i} ({\mbzeta_k \mbzeta_k ^\top | \mathbb{Y}_{T}) }  \label{equation61}\\
 \Theta_2 (\phi_k)  & =  \mathbb{E}_{\hat\phi^i} ({\mbzeta_k \mbz_k ^\top | \mathbb{Y}_{T}) }\label{equation62}\\
 \Theta_3 (\phi_k) & =   \mathbb{E}_{\hat\phi^i} ({\mbz_k \mbz_k ^\top | \mathbb{Y}_{T}) }  \label{equation63} \end{align}
for $\mbzeta_k = \col (\mbs_{k+1} ,\; y_k)$ and $\mbz_k = \col (\mathbf{s}_k, \mathbf{a}_k )$.

\medskip
N.B. The terms $\Theta_1, \Theta_2, \Theta_3$ depend on  $\phi_k$  due to \eqref{control}.

%
%
 \end{lemma}
 \begin{proof}
To begin with, the application of Bayes' rule and leveraging the time varying dynamic model   \eqref{ltv_eq}, one can express the mixture likelihood function as follows,
\begin{align}
 \mathcal{L} (\phi, \hat{\phi^i}) &=  \mathbb{E}_{p_{\hat{\phi}^i} (\mathbb{S}_{T+1}|\mathbb{Y}_T)} (  \log  p_\phi (\mathbb{S}_{T+1}, \mathbb{Y}_T) ) \nonumber \\
&=  \mathbb{E}_{p_{\hat{\phi}^i} (\mathbb{S}_{T+1}|\mathbb{Y}_T)} \log( p(\mathbf{s}_1) \prod_{k=1}^{T} p_{\phi_k} (\mathbf{s}_{k+1},y_k|\mathbf{s}_k)) \nonumber \\
 &= \log p(\mathbf{s}_1) + \sum_{k=1}^T \mathbb{E}_{p_{\hat{\phi}^i} (\mathbf{s}_k|\mathbb{Y}_T)  }  \log   p_{\phi_k} (\mathbf{s}_{k+1}, y_k | \mathbf{s}_k) \label{lbar}  
 \end{align}
 which is \eqref{calLbarcalL} with 
 \begin{align}  \label{barcalL1}
 \bar{\mathcal{L}}_k (\phi_k, \hat{\phi}^i) =
 \mathbb{E}_{p_{\hat{\phi}^i} (\mathbf{s}_k |\mathbb{Y}_T)  }  \log   p_{\phi_k} (\mathbf{s}_{k+1}, y_k | \mathbf{s}_k).
  \end{align}
The expression of $\log p(\mathbf{s}_1)$ given in \eqref{ltv_surrogate} is straightforward 
by using the log of Gaussian p.d.f. of the initial state $\mathbf{s}_1$.
Again, using the log of a Gaussian p.d.f. in \eqref{barcalL1} gives 
 \begin{align*}  
 & -2 \bar{\mathcal{L}}_k (\phi_k, \hat{\phi}^i) \\
 = &    \log |\pmb{\Sigma}^o_k|   + \mathbb{E}_{p_{\hat{\phi}^i} (\mathbf{s}_k |\mathbb{Y}_T)  } \nonumber\\
     &    \Big(\begin{bmatrix}
     \mathbf{s}_{k+1} \\
     y_k
     \end{bmatrix} - \mathbf{A}^o_k \begin{bmatrix}
     \mathbf{s}_{k} \\
     \mathbf{a}_k
     \end{bmatrix}
     \Big)^\top {\pmb{\Sigma}^o_k}^{-1} \Big(\begin{bmatrix}
     \mathbf{s}_{k+1} \\
     y_k
     \end{bmatrix} - \mathbf{A}^o_k \begin{bmatrix}
     \mathbf{s}_{k} \\
     \mathbf{a}_k
     \end{bmatrix}
     \Big) \\
     = &   \mathbb{E}_{p_{\hat{\phi}^i} (\mathbf{s}_k |\mathbb{Y}_T)  }  \Tr  [ {\pmb{\Sigma}^o_k}^{-1} (\mbzeta_k- \mathbf{A}^o_k \mbz_k) (\mbzeta_k- \mathbf{A}^o_k \mbz_k)^\top ]   +  \log |\pmb{\Sigma}^o_k| , \end{align*}
 which matches the expression given in \eqref{ltv_surrogate}.
 The lemma is thus proved.   \end{proof} 
  

  More specifically,  the terms $\Theta_1(\phi_k)$, $\Theta_2(\phi_k)$, and $ \Theta_3(\phi_k)$
can be derived in a straightforward manner. One can refer to Appendix \ref{appednixB} for more explicit details. 
It is noted that they
are composed of elements which can be evaluated from 
 \begin{align} \label{smoothed_entities}
 \mathbb{E}_{\hat{\phi}^i} (\mathbf{s}_k|\mathbb{Y}_T), \; \mathbb{E}_{\hat{\phi}^i} (\mathbf{s}_k \mathbf{s}_k^\top|\mathbb{Y}_T),\;
  \mathbb{E}_{ \hat{\phi}^i} (\textbf{s}_{k+1} \textbf{s}_k^\top|\mathbb{Y}_T).
 \end{align}
In order to evaluate the above mentioned terms, one can take advantage of time-varying linear Kalman filter and R.T.S. smoother components that are introduced below.  Readers are referred to \cite{jazwinski2007stochastic} (Pages 201 - 217) for more details about the procedure. However one cannot use the standard version of filtering and smoothing because in our case where the control action has Gaussian noise, therefore one has to augment the state space modeling to incorporate the covariance of the control action as well.    Specifically,  the Kalman filter equations after augmentation are shown below
where the inputs are $y_k$, $\mbF_k$, $\mbe_k$, and $\mbSigma_k$, for $k=1,2,.., T$, noting the definitions of  $\mathbb{Y}_T$ and $\phi$.
At each iteration, it is implemented with $\phi = \hat{\phi}^i$.
 
 For $k=1,2,.., T$, the time-varying Kalman filter equations (with initialization $\check{\mbs}_{1|1}=\mbs_{1}$ and $\check{\mathbf{P}}_{1|1}={\mathbf{P}}_{1}$) are
\begin{align*}
    {\check{\mathbf{s}}}_{k+1|k} &=  {\widetilde{\mathbf{A}}^d_{k}} \check{\mathbf{s}}_{k|k} + {\mathbf{B}_{k}^d}  { \mbe_{k}  } , \\
     \check{\mathbf{P}}_{k+1|k}&={\widetilde{\mathbf{A}}^d_{k}} \check{\mathbf{P}}_{k|k} ({\widetilde{\mathbf{A}}^d_{k}})^{\top} + \widetilde{{\pmb{\Sigma}}}_{k}^d  ,   \\
  \check{\mathbf{K}}_{k+1} &= \check{\mathbf{P}}_{k+1|k} ( {\widetilde{\mathbf{A}}^r_{k}} )^{\top} ({\widetilde{\mathbf{A}}^r_{k}} \check{\mathbf{P}}_{k+1|k} ({\widetilde{\mathbf{A}}^r_{k}})^{\top} + {\widetilde{\pmb{\Sigma}}_{k}^r})^{-1} , \\
     \check{\mathbf{P}}_{k+1|k+1} &={ \check{ \mathbf{P} }}_{k+1|k} - \mathbf{\check{K}}_{k+1} {\widetilde{\mathbf{A}}^r_{k}} {\check{ \mathbf{P} } }_{k+1|k},\\
  \check{\mathbf{s}}_{k+1|k+1}&={\check{\mathbf{s}}_{k+1|k}} + \check{\mathbf{K}}_{k+1} (y_{k} -{\widetilde{\mathbf{A}}^r_{k}} \check{\mathbf{s}}_{k+1|k}- {\mathbf{B}^r_{k}} { \mbe_{k}} ) ,
\end{align*}
where ,
\begin{align*}
    {\widetilde{\mathbf{A}}^r_{k}} &= \mathbf{A}^r_{k} + \mathbf{B}^r_{k} { \mathbf{F}_{k}}, \\
    {\widetilde{\mathbf{A}}^d_{k}}& = \mathbf{A}^d_{k} + \mathbf{B}^d_{k} { \mathbf{F}_{k}} ,\\
    \widetilde{{\pmb{\Sigma}}}_{k}^d& = \mathbf{B}^d_{k} \mbSigma_{k} {\mathbf{B}^d_{k}} ^\top +\mbSigma^d_{k} ,\\
   \widetilde{{\pmb{\Sigma}}}_{k}^r &= \mathbf{B}^r_{k} \mbSigma_{k}  {\mathbf{B}^r_{k}} ^\top +\mbSigma^r_{k}.
\end{align*}
The time-varying recursive smoother equations are given as follows, for { $k=T, T-1,\cdots,1$}, 
with $\hat{\mathbf{s}}_{T+1|T} = \check{\mathbf{s}}_{T+1|T} $ and $\hat{\mathbf{P}}_{T+1|T} =\check{\mathbf{P}}_{T+1|T}$,
\begin{align*}
        \mathbf{J}_k &=\check{\mathbf{P}}_{k|k} {\widetilde{\mathbf{A}}^d_{k}} (\check{\mathbf{P}}_{k+1|k})^{-1} , \\
      \hat{\mathbf{s}}_{k|T} &= \check{\mathbf{s}}_{k|k} + \mathbf{J}_{k} (\hat{\mathbf{s}}_{k+1|T} -       \check{\mbs}_{k+1|k} )  , \\
    \hat{\mathbf{P}}_{k|T} &=\check{\mathbf{P}}_{k|k} + \mathbf{J}_k (\hat{\mathbf{P}}_{k+1|T}- \check{\mathbf{P}}_{k+1|k} ) \mathbf{J}_k^{\top} .
\end{align*}
One can calculate the one-lag smoothed term $ \hat{\textbf{M}}_{k|T} $ backwards
with the initialization and the iteration with $k=T,T-1,\cdots,2$, as follows,  
\begin{align*}
   \hat{\mathbf{M}}_{T+1|T} & = (\mathbf{I} - \mathbf{\check{K}}_{T+1} \widetilde{\mathbf{{A}}}_T^r ) \widetilde{\mathbf{A}}_T^d \check{\mathbf{P}}_{T|T}  \\
\hat{\mathbf{M}}_{k|T} &= \check{\mathbf{P}}_{k|k} \mathbf{J}_{k-1}^{\top}+ \mathbf{J}_k (\hat{\mathbf{M}}_{k+1|T} -  \widetilde{\mathbf{A}}_{k}^d \check{\mathbf{P}}_{k|k})\mathbf{J}_{k-1}^{\top} .\end{align*}

After  evaluating the filtered, smoothed estimates of states, the error covariance matrices and one lag covariance matrices for all time steps, one can evaluate the terms of  \eqref{smoothed_entities} as follows,
\begin{align}
       \mathbb{E}_{\hat{\phi}^i} (\mathbf{s}_k|\mathbb{Y}_T) &= \hat{\mathbf{s}}_{k|T}  \\
\mathbb{E}_{\hat{\phi}^i }  (\mathbf{s}_k \mathbf{s}_k^\top|\mathbb{Y}_T ) & =  \hat{\mathbf{s}}_{k|T} \hat{\mathbf{s}}_{k|T}^\top + \hat{\mathbf{P}}_{k|T} \triangleq \mathbf{G}_k  , \label{eq:G_k_def}\\
 \mathbb{E}_{ \hat{\phi}^i }(\textbf{s}_{k+1} \textbf{s}_k^\top|\mathbb{Y}_T) & = \hat{\mathbf{s}}_{k+1|T} \hat{\mathbf{s}}_{k|T} + \hat{\mathbf{M}}_{k+1 |T} \triangleq \mathbf{M}_{k+1|T} .
\end{align}
It is noted that $\hat{\mathbf{P}}_{k|T} > 0$; see \cite{gibson2005robust}-Lemma C.4.

%

\subsection{A practical algorithm for maximization of mixture likelihood} \label{sec:practical}

The attention is now turned towards maximization of the mixture likelihood 
$\mathcal{L} (\phi,\hat{\phi}^i)$, called the M-step in the EM architecture. 
 The proposed optimization paradigm seeks a better policy parameter $\phi =\hat\phi^{i+1}$ for the next iteration
 than $\phi =\hat\phi^{i}$ in the sense of maximizing (or increasing)
 $\mathcal{L} (\phi,\hat\phi^i)$. 
 From  Lemma~\ref{lemma:LTV}, one has 
 \begin{align} \label{optiphi}
\hat{ \phi}^{i*}    = \arg\max_{\phi} \mathcal{L} (\phi, \hat\phi^i)
= \arg\max_{\phi} \sum_{k=1}^T  \bar{\mathcal{L}}_k (\phi_k, \hat{\phi}^i).
\end{align}
So, it is ideal to select $\hat\phi^{i+1} = \hat{ \phi}^{i*} $.
 
However, it is typically difficult to compute the optimal $\hat{ \phi}^{i*}$
over the entire sequence of control action $\{\mathbf{a}_1, \mathbf{a}_2,.., \mathbf{a}_T\}$.
 The principle of EM as optimal control reduces the maximization of ${\mathcal{L}} (\phi,\hat{\phi^i})$ over the entire sequence of control action for each time step.  In other words, one tends to maximize ${\mathcal{L}} (\phi,\hat{\phi^i})$ for each time step according to the iterative procedure below. For $j=1,..,T$, we solve the local optimization problem recursively,
 \begin{align} \label{optiphij}
\hat{ \phi}_1^{i*}   
& =  \arg\max_{\phi_1  } \sum_{k=1}^T  \bar{\mathcal{L}}_k (\phi_1, \hat{\phi}^i) \nonumber\\
\hat{ \phi}_j^{i*}   
& =  \arg\max_{\phi_j   }   \sum_{k=1}^{T} \bar{\mathcal{L}}_k (\phi_j, \col (\hat{\phi}_1^{i*}, \cdots, \hat{\phi}_{j-1}^{i*}, 
\hat{\phi}_{j}^{i} ,\cdots, \hat{\phi}_{T}^{i}) )  , \nonumber\\
& j=2,\cdots, T
\end{align}
Then, a better policy parameter for the next iteration is selected as 
 $\hat{ \phi}^{i+1} = \col(\hat{ \phi}^{i*}_1, \cdots, \hat{ \phi}^{i*}_T) \approx \hat{ \phi}^{i*} $.
Obviously, the dimension of the optimization problem of $\phi_j$ in \eqref{optiphij}, for $j=1,\cdots, T$,
is significantly lower than that for $\phi$ in \eqref{optiphi}.
 For brevity we  would refer the above optimization problem as \textbf{SOC-EM I} in the subsequent part of the paper.
  
The optimal controller parameters of \cite{levine2016end,tassa2012synthesis} are dependent on time instant and as we utilize a similar framework, therefore  we also exploit the time dependent nature of control law. While implementing our methodology into practice the optimization routine  \eqref{optiphij} suffers from intensive nature of computational costs.  Therefore, we try increase the speed of parameter search by converting the mixture likelihood into a surrogate function which is mathematically cheaper and tractable to evaluate.
The modified optimization for each optimization instance  is
\begin{align} \label{optim2}
  \hat{\phi}^{i*}_j  = \arg\max_{\phi_j}  \sum_{k=1}^T  \bar{\mathcal{L}}_k (\phi_j,\hat\phi^i),\;
  j=1,\cdots,T,
\end{align}
for searching $\phi_j$ in a neighborhood of $\hat\phi_j^i$.
Then, a policy parameter for the next iteration is selected as 
$\hat{ \phi}^{i+1} = \col(\hat{ \phi}^{i*}_1, \cdots, \hat{ \phi}^{i*}_T) \approx \hat{ \phi}^{i*} $.
We would refer to this routine as \textbf{SOC-EM II}  in the sequel.
{\begin{remark}	 
The optimization problem \eqref{optiphi} is in the parameter space of a very high dimension of $T(n_s n_a + n_a^2+ n_a )$,
which motivates the practical approximation  \eqref{optiphij} and  \eqref{optim2} that are carried out with respect to an individual time step rather than all of $T$ steps in one go. After convergence of the optimization routine in one time step, the result is utilized in the immediately next one.  The approximation of SOC-EM I and  SOC-EM II is a heuristic approach to deal with heavy computational expense. There is a trade-off that undoubtedly needs to be made between complexity of the algorithm and the quality of the solutions. 
It has been tested on extensive experiments  that SOC-EM I and SOC-EM II have  similar performance and both of them 
demonstrate satisfactory performance in terms of the metrics described later in Section~\ref{sec:results}.
It is worth mentioning that another advantage of the approximation \eqref{optim2} is that it allows parallel computation with multiple CPUs, which tremendously reduces the computational time.  
\end{remark}
}

 It is worth mentioning that initializing EM with a considerably good parameter vector $\hat{\phi}^0$  is critical. For example, theoretical studies by \cite{baudry2015mixtures,mclachlan2007algorithm} revealed that the convergence of EM algorithm is highly dependent on the parameters with which it is initialized. In order to address the initialization issue, we employ the parameters of the well established  trajectory optimization strategies. In this paper, we specifically utilize  differential dynamic programming   based optimal control methods such as 1) iLQG (\cite{tassa2012synthesis}); 2) MPC (\cite{zhang2016learning}), and 3) BADMM (\cite{montgomery2016guided}) to carry forward the optimization routine. 
These three techniques will be called as the  \enquote{\textbf{baselines}} in the sequel.

{Now, the overall SOC-EM algorithm is summarized in Algorithm~\ref{alg:EM}.
The number of recursion in the algorithm is specified a priori. 
It is noted that every recursion starts with line 2 for model improvement with 
a new controller parameter. In practice, after a certain amount of recursions,
there is no more significant model improvement, therefore ``go to 2"  can be replaced
by ``go to 3" to skip  line 2.}  While carrying out the optimization, 
 $(\pmb{\hat{\Sigma}}^a_j)^{i*},\;  j = 1,2,\cdots, T$, the covariance matrix  component of $\hat{\phi}^{i*}_j$,  
 might (or certainly) loose its positive definiteness property, so in order to preserve it, we adopt {the approach originally} proposed by {\cite{trove.nla.gov.au/work/6003470}}. {In this strategy} instead of propagating $\pmb{\hat{\Sigma}}^a_k$, {one} propagates its square root, ${\pmb{\hat{\Sigma}}^a_k}^{\frac{1}{2}}$, i.e., 
  $ {\pmb{\hat{\Sigma}}^a_k}  = ({\pmb{\hat{\Sigma}}^a_k}^{\frac{1}{2}} ) ^\top {\pmb{\hat{\Sigma}}^a_k}^{\frac{1}{2}}$, 
  by carrying out a Cholesky decomposition before optimization. In order to save computational time, we also propagated the  square-roots of the filtered $\mathbf{\check{P}}_{k+1|k}$, $\mathbf{\check{P}}_{k+1|k+1}$ and smoothed $\mathbf{\hat{P}}_{k|T}$.

\floatname{algorithm}{Algorithm}
{ 
 \begin{algorithm}[t]
\caption{The SOC-EM Algorithm}\label{alg:EM}
\begin{algorithmic}[1]
		\STATE \textit{(Initialization)} Let $i=0$;  initialize $\hat{\phi}^0$ from one of the  baselines (iLQG/MPC/BADMM).
  
\STATE   Run the real system under the controller \eqref{control}  with the controller parameter vector $\hat{\phi}^i$ to collect the data set  $\mathcal{ D}$; 
identify the dynamic model \eqref{ltv_eq} by fitting it to $\mathcal{ D}$ via VB inference.

	\STATE  Generate the cost observations ${\mathbb{Y}_{T}}$  using
the dynamic model \eqref{ltv_eq}   
and the controller  \eqref{control} with the controller parameter vector $\hat\phi^i$.

	\STATE  Perform the Kalman filter and R.T.S. smoother recursions to evaluate \eqref{smoothed_entities}
	and hence  $ {\bar{\mathcal{L}} _k(\phi_k,\hat{\phi^i})}$, $k=1,\cdots, T$.
	 
	 \STATE Find $\hat{ \phi}_j^{i*} $, $j=1,\cdots T$ from SOC-EM I \eqref{optiphij} or SOC-EM II \eqref{optim2}. 
	 
	 \STATE  Let $\hat{ \phi}^{i*} = \col(\hat{ \phi}^{i*}_1, \cdots, \hat{ \phi}^{i*}_T)$ be the approximate solution 
	 to \eqref{valuefn_ltv2}.
	 
	 \STATE Update  $\hat{ \phi}^{i+1} = \hat{ \phi}^{i*} $. 
	 
	 \IF{$i+1 <$ number of recursion}
	 \STATE   Let $i=i+1$; go to 2.
	 \ENDIF
	 
%
%
	\STATE \textbf{Return} The controller parameter vector $\hat{\phi}^{i+1}$.
\end{algorithmic}
\end{algorithm}
}

%

\subsection{Uniqueness of controller parameter estimation}

The two theorems in this subsection exploit the closed form nature of the gradient and  the Hessian of the mixture log likelihood 
to deliver a theoretical proof of the uniqueness of solution to the optimization problem \eqref{optim2}, i.e.,  \textbf{SOC-EM II}. 
It is easy to verify that the theorems still hold  with $\hat\phi_k^i$ in  \eqref{jacobian_set_zero} and 
\eqref{total_matris_krace} replaced by $\hat\phi_k^{i*}$ for $k=1,\cdots, j-1$ and hence guarantee the 
the uniqueness of solution to the optimization problem \eqref{optiphij},  \textbf{SOC-EM I}. 

\begin{theorem} \label{uniqueness1} 
For the function ${\bar{\mathcal{L}}_k (\phi_k,\hat{\phi^i})}$ defined in \eqref{ltv_surrogate},  the following equation,  
\begin{align} \label{jacobian_set_zero}
    \nabla_{\phi_j} \Big\{  \sum_{k=1}^T  \Bar{\mathcal{L}}_k(\phi_j,\hat\phi^i) \Big\}   = \mathbf{0}, \; j=1,\cdots, T,
\end{align}
has a unique solution for any given parameter  $\hat{\phi}^i$.
\end{theorem}

\begin{proof}  Recall that the function ${\bar{\mathcal{L}}_k (\phi_k,\hat{\phi^i})}$ defined in \eqref{ltv_surrogate} is expressed in terms of
$\Theta_1(\phi_k), \Theta_2(\phi_k), \Theta_3(\phi_k)$ in \eqref{equation61}-\eqref{equation63}. The terms are composed of
\begin{align*}
    \mathbb{E}_{{\hat{\phi}^i} } (y_k y_k^{\top}| {\mathbb{Y}_T}) ,\; \mathbb{E}_{{\hat{\phi}^i}} (\mathbf{a}_k \mathbf{a}_k^{\top}| {\mathbb{Y}_T}), \; \mathbb{E}_{{\hat{\phi}^i} } (y_k \mathbf{a}_k^{\top}| {\mathbb{Y}_T}),
\end{align*}  
where $\mathbf{a}_k$ explicitly depends on ${\phi}_k$.
From the detailed expression given in Appendix \ref{Jacs}, one has
\begin{align}
 \nabla_{\phi_j} \Big\{  \sum_{k=1}^T  \Bar{\mathcal{L}}_k(\phi_j,\hat\phi^i) \Big\}   
 =  \sum_{k=1}^T \nabla_{\phi_j } \Tr {\pmb{\lambda}_k}  + {\pmb{\mathcal{O}}}^\top    \label{equation_79}
\end{align}
where 
\begin{align}
 {\pmb{\lambda}_k}  = {\pmb{\Sigma}^d_k}^{-1} \mathbf{B}^d_k \mathbb{E}_{\hat{\phi}^i} \big( \mathbf{a}_k \mathbf{a}_k^\top |\mathbb{Y}_T  \big) {\mathbf{B}^d_k}^\top     
\end{align}
and ${\pmb{\mathcal{O}}}$  represents some constant column vector,  independent of ${\phi}_j =\col (\mbf_j, \mbe_j, \mbsigma_j)$.
In particular,  one has ${\pmb{\lambda}_k}= \pmb{\lambda}_{1,k}+\pmb{\lambda}_{2,k}+\pmb{\lambda}_{3,k} + \pmb{\lambda}_{4,k}+\pmb{\lambda}_{5,k}  $  with
  \begin{align*}
 \pmb{\lambda}_{1,k}&=(\pmb{\Sigma}^d_k)^{-1} \mathbf{B}^d_k \mathbf{F}_j \mathbf{G}_k \mathbf{F}_j^\top  {\mathbf{B}^d_k}^\top \\
 \pmb{\lambda}_{2,k}&=(\pmb{\Sigma}^d_k)^{-1} \mathbf{B}^d_k \mbe_j \mbe_j^\top {\mathbf{B}^d_k}^\top\\
 \pmb{\lambda}_{3,k}&=  (\pmb{\Sigma}^d_k)^{-1} \mathbf{B}^d_k \mbSigma_j {\mathbf{B}^d_k}^\top \\
 \pmb{\lambda}_{4,k}&= ({\pmb{\Sigma}^d_k})^{-1} \mathbf{B}^d_k \mathbf{F}_j \hat{\mathbf{s}}_{k|T} \mbe_j^\top {\mathbf{B}^d_k}^\top  \\
\pmb{\lambda}_{5,k}&= ({\pmb{\Sigma}^d_k})^{-1} \mathbf{B}^d_k \mbe_j  \hat{\mathbf{s}}_{k|T}^\top \mathbf{F}_j^\top {\mathbf{B}^d_k}^\top .
\end{align*}
Also,  ${\pmb{\mathcal{O}}}$ is of the special structure
\begin{align} \label{Ostructure}
 {\pmb{\mathcal{O}}} =   \begin{bmatrix}
   \sum_{k=1}^T \pmb{\mathcal{O}}_{1,k}\\
 \sum_{k=1}^T \pmb{\mathcal{O}}_{2,k}\\
 \mathbf{0}
   \end{bmatrix} , \end{align} 
with the dimensions of $\pmb{\mathcal{O}}_{1,k}$, $ \pmb{\mathcal{O}}_{2,k}$
and $\mathbf{0}$ corresponding to those of ${\mbf_j}$, ${\mbe_j}$ and  
$\mbsigma_j$, respectively. The explicit expression of  
$ \pmb{\mathcal{O}}_{1,k}$ and $ \pmb{\mathcal{O}}_{2,k}$ can be obtained from the equations in Appendix~\ref{Jacs}.

Below, we calculate the derivative of the terms in \eqref{equation_79} with respect to 
$\mbf_j$, $\mbe_j$, and $\mbsigma_j$, respectively. 
 
Firstly, with respect to  $\mbf_j$, one has 
\begin{align*}
      \nabla_{\mbf_j}   \Tr  \pmb{\lambda}_{1,k} =\mbf_j^\top  \pmb{\mathcal{Z}} ^{1}_k ,\;
      \nabla_{\mbf_j}  \Tr  \pmb{\lambda}_{2,k} =  0,\;
      \nabla_{\mbf_j}    \Tr  \pmb{\lambda}_{3,k} = 0
\end{align*}
and hence 
\begin{align*}
 & \nabla_{\mbf_j}  \sum_{k=1}^T    \Tr  \pmb{\lambda}_{4,k} \\
 &=\nabla_{\mbf_j}  \sum_{k=1}^T   \Tr  \Big\{ ({\pmb{\Sigma}^d_k})^{-1} \mathbf{B}^d_k \mathbf{F}_j \hat{\mathbf{s}}_{k|T} \mbe_j^\top {\mathbf{B}^d_k}^\top \Big\} \\ 
 &=\nabla_{\mbf_j}  \sum_{k=1}^T    {\Tr  \Big\{ (\mbe_j^\top {\mathbf{B}^d_k}^\top) ({\pmb{\Sigma}^d_k}^{-1} \mathbf{B}^d_k \mathbf{F}_j \hat{\mathbf{s}}_{k|T}  ) \Big\} } \\
 &=  (\mathbf{B}^d_k \mbe_j)^\top \nabla_{\mbf_j}  (\hat{\mathbf{s}}_{k|T}^\top \otimes { \pmb{\Sigma}_k^d}^{-1} \mathbf{B}^d_k)
 \mbf_j \\
   &  =   [(\hat{\mathbf{s}}_{k|T} \otimes {  {\mathbf{B}^d_k}^\top {\pmb{\Sigma}_k^d}^{-1}}^\top )  (\mathbf{I} \otimes {\mathbf{B}^d_k} )
   \mbe_j ]^\top \nonumber 
\end{align*}
and, similarly, 
\begin{align*}
    \nabla_{\mbf_j}    \Tr  \pmb{\lambda}_{5,k}  = [ (\hat{\mathbf{s}}_{k|T} \otimes {\mathbf{B}^d_k}^\top )  (\mathbf{I} \otimes  {  {\pmb{\Sigma}_k^d}^{-1}}  {{\mathbf{B}^d_k}}  )  \mbe_j]^\top.
\end{align*}
Here,   $\pmb{\mathcal{Z}} ^1_k = \pmb{\mathcal{Z}}_k^{1,0} +  \pmb{\mathcal{Z}} ^{1,1}_k$ with
\begin{align*}
   \pmb{\mathcal{Z}} _k^{1,0}   &= 2  \hat{\mathbf{s}}_{k|T} \hat{\mathbf{s}}_{k|T}^\top  \otimes  {\mathbf{B}^d_k}^\top {\pmb{\Sigma}^d_k}^{-1} \mathbf{B}^d_k   \\
   \pmb{\mathcal{Z}} _k^{1,1} &= 2  \hat{\mathbf{P}}_{k|T}   \otimes  {\mathbf{B}^d_k}^\top {\pmb{\Sigma}^d_k}^{-1} \mathbf{B}^d_k  .
\end{align*}

Secondly, with respect to  $\mbe_j$, one has 
\begin{align*}
\nabla_{\mbe_j}   \Tr  \pmb{\lambda}_{1,k} &= 0\\
\nabla_{\mbe_j}   \Tr  \pmb{\lambda}_{2,k} &= \mbe_j^\top \pmb{\mathcal{Z}} ^{2}_k \\
\nabla_{\mbe_j}   \Tr  \pmb{\lambda}_{3,k} &=   0 \\
\nabla_{\mbe_j}  \Tr  \pmb{\lambda}_{4,k} &= 
 \mbf_j^\top (\mathbf{I} \otimes {\mathbf{B}^d_k}^\top) (\hat{\mathbf{s}}_{k|T} \otimes { \pmb{\Sigma}_k^d}^{-1} \mathbf{B}^d_k ) \\
\nabla_{\mbe_j}   \Tr  \pmb{\lambda}_{5,k}  &= \mbf_j^\top (\mathbf{I} \otimes {\pmb{\Sigma}_k^d}^{-1} {\mathbf{B}^d_k})^\top (\hat{\mathbf{s}}_{k|T} \otimes \mathbf{B}^d_k ) 
\end{align*}
with
\begin{align*}
      \pmb{\mathcal{Z}} ^2_k  &=2    \mathbf{I} \otimes  {\mathbf{B}^d_k}^\top {\pmb{\Sigma}^d_k}^{-1} \mathbf{B}^d_k   . 
 \end{align*}
   
Thirdly, with respect to  $\mbsigma_j$, the only nonzero derivative is
\begin{align}
    \nabla_{\mbsigma_j} \Tr   \pmb{\lambda}_{3,k} =
    \mbsigma_j^\top   \pmb{\mathcal{Z}}^2_k, 
\end{align}
using  the equations in Appendix~\ref{Jacs}.

From above, the equation \eqref{jacobian_set_zero} is equivalent to    
   \begin{align} \label{optimal_estimates}
    \nabla_{\phi_j} \Big\{  \sum_{k=1}^T  \Bar{\mathcal{L}}_k (\phi_j,\hat\phi^i) \Big\}  =
\phi_j^\top   \begin{bmatrix}
   \pmb{\mathcal{Z}}  & \mathbf{0}\\
\mathbf{0} & \pmb{\mathcal{Z}}^2 
   \end{bmatrix}        +   \pmb{\mathcal{O}}^\top  = \mathbf{0}
\end{align}
where   $\pmb{\mathcal{Z}}^2 = \sum_{k=1}^T \pmb{\mathcal{Z}}^2_k$, $\pmb{\mathcal{Z}}  =\sum_{k=1}^T \pmb{\mathcal{Z}}_k$, 
$\pmb{\mathcal{Z}}_k=  \begin{bmatrix}
  \pmb{\mathcal{Z}}_k^1   &  \pmb{\mathcal{Z}}_k^3  \\
 {\pmb{\mathcal{Z}}_k^3}^\top &    {\pmb{\mathcal{Z}}_k^2 }
   \end{bmatrix}$ with \begin{align*}
   \pmb{\mathcal{Z}} ^3_k  =& 2 (  \mathbf{I}_{1} \otimes {\mathbf{B}^d_k}^\top) (\hat{\mathbf{s}}_{k|T}^\top \otimes {\pmb{\Sigma}^d_k}^{-1} {\mathbf{B}^d_k}).
      \end{align*}

 What is left is to prove the existence of a unique solution $\phi_j$ to the equation
 \eqref{optimal_estimates}. It suffices to show that  $\pmb{\mathcal{Z}}  >0$
 and  $ \pmb{\mathcal{Z}}^2>0$, or $\pmb{\mathcal{Z}}_k >0$
 and  $ \pmb{\mathcal{Z}}^2_k>0$.

 Since the matrix $\mathbf{B}^d_k$ has a  full column rank, $\pmb{\Sigma}^d_k >0$ and $\hat{\mathbf{P}}_{k|T} > 0$, one has
 $\pmb{\mathcal{Z}} ^2_k  >0$ and  $\pmb{\mathcal{Z}} _k^{1,1} >0$. 
Next,   the decomposition of the matrix $\pmb{\mathcal{Z}}_k$ gives
\begin{align} \label{decomposition}
     \pmb{\mathcal{Z}}_k  = & \begin{bmatrix}
   \mathbf{I} &  \pmb{\mathcal{Z}} ^3_k  {\pmb{\mathcal{Z}} ^2_k}^{-1} \\
   \mathbf{0} & \mathbf{I}
    \end{bmatrix} 
    \begin{bmatrix}
   \pmb{\mathcal{Z}} _k^{1}    -  \pmb{\mathcal{Z}} _k^3 {\pmb{\mathcal{Z}} _k^2}^{-1}  {\pmb{\mathcal{Z}} _k^3}^\top   &  \mathbf{0} \\
   \mathbf{0} &  \pmb{\mathcal{Z}} _k^2
   \end{bmatrix} \nonumber\\
   &  
    \begin{bmatrix}
   \mathbf{I} & \mathbf{0}  \\
 {\pmb{\mathcal{Z}} ^2_k}^{-1} {\pmb{\mathcal{Z}} ^3_k}^\top   & \mathbf{I}
    \end{bmatrix}.
   \end{align}
It is noted that $\pmb{\mathcal{Z}}_k ^{1,0} = \pmb{\mathcal{Z}}_k^3 {\pmb{\mathcal{Z}}_k^2}^{-1}  {\pmb{\mathcal{Z}}^3_k}^\top $, 
which implies
      \begin{align*}
 \pmb{\mathcal{Z}} _k^{1}    -  \pmb{\mathcal{Z}} _k^3 {\pmb{\mathcal{Z}} _k^2}^{-1}  {\pmb{\mathcal{Z}} _k^3}^\top
 =\pmb{\mathcal{Z}}_k ^{1,1} >0
      \end{align*}
 and hence $\pmb{\mathcal{Z}}_k >0$. The proof is thus completed.
\end{proof}

\begin{remark}  \label{remark-zerocovariance}
The unique solution $\phi_j$ to \eqref{optimal_estimates} is 
   \begin{align} 
    \hat\phi^{i*}_j  =-  \begin{bmatrix}
   \pmb{\mathcal{Z}}  & \mathbf{0}\\
\mathbf{0} & \pmb{\mathcal{Z}}^2 
   \end{bmatrix}  ^{-\top}  \pmb{\mathcal{O}}.    
\end{align}
Denote
$\hat\phi_j^{i*} =\col( \hat\mbf_j^{i*},  \hat\mbe_j^{i*},  \hat\mbsigma_j^{i*})$,
whose covariance matrix component  is
 $\hat\mbsigma_j^{i*}=\mathbf{0}$, due to
 \eqref{Ostructure}. However, in practical scenarios if we employ an optimization routine, then it would tend to decrease towards zero in an iterative manner which can also be validated from the simulation results. \end{remark}

The following theorem  addresses the  positive definiteness property of the negative Hessian of the surrogate function  of mixture log likelihood.

\begin{theorem} \label{theorem:global_maximizer}
For the function ${\bar{\mathcal{L}} (\phi_k,\hat{\phi}^i_k)}$ defined in \eqref{ltv_surrogate},  the following inequality  
\begin{align}
& - \nabla^2_{\phi_j} \Big\{  \sum_{k=1}^T  \bar{\mathcal{L}}_k (\phi_j,\hat{\phi}^i) \Big\}    > 0, \; j=1,\cdots, T  \label{total_matris_krace}
\end{align}
always holds for any given parameter $\hat{\phi^i}$. 
\end{theorem}

\begin{proof}
Following the proof of Theorem~\ref{uniqueness1},
one has
\begin{align}
 & - \nabla^2_{\phi_j} \Big\{  \bar{\mathcal{L}}_k (\phi_j,\hat\phi^i) \Big\}  =
        - \nabla^2_{\phi_j} \Big\{  \Tr  {\pmb{\lambda}_k}    \Big\}  \nonumber\\
   = & \begin{bmatrix}
    \frac{-\partial (\cdot) }{ \partial \mbf_j \partial \mbf_j^\top} &    \frac{-\partial (\cdot) }{ \partial \mbf_j \partial \mbe_j^\top} &    \frac{-\partial (\cdot) }{ \partial \mbf_j \partial \mbsigma_j^\top} \\
      \frac{-\partial (\cdot) }{ \partial \mbe_j \partial \mbf_j^\top} &   \frac{-\partial (\cdot) }{ \partial \mbe_j \partial \mbe_j^\top} &   \frac{-\partial (\cdot) }{ \partial \mbe_j \partial \mbsigma_j^\top} \\ 
      \frac{-\partial (\cdot) }{ \partial \mbsigma_j \partial \mbf_j^\top}  &   \frac{-\partial (\cdot) }{ \partial \mbsigma_j \partial \mbe_j^\top}   &   \frac{-\partial (\cdot) }{ \partial \mbsigma_j \partial \mbsigma_j^\top} \\
  \end{bmatrix}  \Tr   {\pmb{\lambda}_k}. \label{Hessianpf1}
\end{align} 

Firstly, we calculate the three diagonal elements of \eqref{Hessianpf1} below, in \eqref{eighteen}, \eqref{diag_2}, and \eqref{diag_3}, respectively. 
The calculation starts from 
\begin{align*}
    & \nabla^{2}_{\mbf_j} \Tr  \pmb{\lambda}_{1,k} \\
    =&  \nabla_{\mbf_j}  \left\{ \frac{\partial  \Tr  \pmb{\lambda}_{1,k}}{\partial \mbf_j^\top} \right\}^\top   \\
   = &   \nabla_{\mbf_j}  \left\{   \nabla_{\text{vec} \{\mathbf{X}\} } \Tr  (\pmb{\lambda}_{1,k}  )  \frac{\partial \{
   \text{vec} (\mathbf{X})\} }{\partial \mbf_j^\top}  \right\}^\top   
   \end{align*} 
for  $\mathbf{X}=\mathbf{B}^d_k \mbF_j$.   By invoking some basic identities,  
the equation continues
with   
   \begin{align*}
   =&  \nabla_{\mbf_j}  \{  ( {\pmb{\Sigma}^d_k}^{-1} \mathbf{X} \mathbf{G}_k + {{\pmb{\Sigma}^d_k}^{-1}}^\top \mathbf{X} \mathbf{G}_k^\top )^\top  (\mathbf{I}_{n_a} \otimes \mathbf{B}^d_k)   \} ^\top \\
    =&   \nabla_{\mbf_j}  \Big\{
  \Big( (  \mathbf{G}_k^\top \otimes {{\pmb{\Sigma}^d_k}^{-1}})
  (\mathbf{I}_{n_a} \otimes \mathbf{B}^d_k )  \mbf_j   \\
  &  +  ( \mathbf{G}_k \otimes {{\pmb{\Sigma}^d_k}^{-1}}^\top)
    (\mathbf{I}_{n_a} \otimes \mathbf{B}^d_k) \mbf_j   \Big)^\top (\mathbf{I}_{n_a} \otimes \mathbf{B}^d_k) \Big\}^\top \\
  =&  \nabla_{\mbf_j}   \{ (\pmb{\varpi}_{k} (\mathbf{I}_{n_a} \otimes \mathbf{B}^d_k ) \mbf_j )^\top (\mathbf{I}_{n_a} \otimes \mathbf{B}^d_k) \}^\top  \\
  =&  \nabla_{\mbf_j}   \{ \mbf_j^\top (\mathbf{I}_{n_a} \otimes {\mathbf{B}^d_k}^\top ) \pmb{\varpi}_{k} {(\mathbf{I}_{n_a} \otimes \mathbf{B}^d_k )} \}^\top  \\
    =& \nabla_{\mbf_j}   \{ \mbf_j^\top     \pmb{\mathcal{Z}} ^1_k     \}^\top 
\end{align*} 
for $ \pmb{\varpi}_{k} = \mathbf{G}_k^\top \otimes {\pmb{\Sigma}^d_k}^{-1}  + \mathbf{G}_k \otimes {\pmb{\Sigma}^d_k}^{-1}$.  
Furthermore, 
    \begin{align}
 \frac{-\partial (\Tr \pmb{\lambda}_{k}) }{ \partial \mbf_j \partial \mbf_j^\top}  =     \frac{-\partial (\Tr \pmb{\lambda}_{1,k}) }{ \partial \mbf_j \partial \mbf_j^\top} =      \nabla^2_{ \mbf_j}   \Tr  \pmb{\lambda}_{1,k}    =  \mathbf{I} \otimes \pmb{\mathcal{Z}} ^1_k .  \label{eighteen}
\end{align} 

Similarly, one has
\begin{align}
 \frac{-\partial (\Tr \pmb{\lambda}_{k}) }{ \partial \mbe_j \partial \mbe_j^\top} &=
 \frac{-\partial (\Tr \pmb{\lambda}_{2,k}) }{ \partial \mbe_j \partial \mbe_j^\top}    = \mathbf{I} \otimes  \pmb{\mathcal{Z}} ^2_k \label{diag_2}\\
    \frac{-\partial (\Tr \pmb{\lambda}_{k}) }{ \partial \mbsigma_j \partial \mbsigma_j^\top} &=
   \frac{-\partial (\Tr \pmb{\lambda}_{3,k}) }{ \partial \mbsigma_j \partial \mbsigma_j^\top}  = \mathbf{I} \otimes  \pmb{\mathcal{Z}} ^2_k. \label{diag_3}
\end{align}

Secondly, the only non-zero off-diagonal element of \eqref{Hessianpf1}  is   
\begin{align}
 &  
  \frac{-\partial ( \Tr  \pmb{\lambda}_{k} ) }{ \partial \mbf_j \partial \mbe_j^\top} =
  \frac{-\partial (\Tr  \pmb{\lambda}_{4,k} + \Tr  \pmb{\lambda}_{5,k} ) }{ \partial \mbf_j \partial \mbe_j^\top} 
  \nonumber \\
  = &   \nabla_{ \mbe_j^\top} \Big\{ [ (\hat{\mathbf{s}}_{k|T} \otimes {\mathbf{B}^d_k}^\top )  (\mathbf{I} \otimes  {  {\pmb{\Sigma}_k^d}^{-1}}  {{\mathbf{B}^d_k}}  )   \mbe_j ]^\top \nonumber \\
 & + 
 [(\hat{\mathbf{s}}_{k|T} \otimes {  {\mathbf{B}^d_k}^\top {\pmb{\Sigma}_k^d}^{-1}}^\top )  (\mathbf{I} \otimes {\mathbf{B}^d_k} ) \mbe_j ]^\top \Big\} = \mathbf{I} \otimes \pmb{\mathcal{Z}} ^3_k. \label{offdiag}
\end{align}

Finally, by combing \eqref{eighteen}, \eqref{diag_2}, \eqref{diag_3}, and \eqref{offdiag},  and noting $\pmb{\mathcal{Z}}_k >0$
 and  $ \pmb{\mathcal{Z}}^2_k>0$ in the proof of Theorem~\ref{uniqueness1},  one can verify 
\begin{align*}
    - \nabla^2_{\phi_j} \Big\{  \bar{\mathcal{L}}_k (\phi_j,\hat\phi^i) \Big\}  
 =  \begin{bmatrix}
       \mathbf{I} \otimes  \pmb{\mathcal{Z}} ^1_k  & \mathbf{I} \otimes \pmb{\mathcal{Z}}^3_k & 0 \\
         (\mathbf{I} \otimes \pmb{\mathcal{Z}}^3_k)^\top  &     \mathbf{I} \otimes  \pmb{\mathcal{Z}} ^2_k   &0  \\
  0 &  0&     \mathbf{I} \otimes  \pmb{\mathcal{Z}} ^2_k    \\
    \end{bmatrix} > 0,
\end{align*}  which implies   \eqref{total_matris_krace}.  
This ends the proof. 
\end{proof}

\subsection{Analysis of control covariance matrix}\label{section:noise}

{Balance between exploration and exploitation is an age-old problem in reinforcement learning. On one hand, exploration means how much state space needs to be explored to gather rich/good data so that one can improve the model and  ultimately the improved model delivers better policies. A covariance matrix in the stochastic control policy is thus indispensable in the exploration stage. 	
This is the reason we do not use a deterministic policy in the first place that does not allow any exploration. 
On the other hand, one needs to probe/search/optimize the explored region and hence reduce  the stochasticity in the control policies caused by exploration,  which is known as exploitation.
In this section, we show some theoretical results which reveal how exploitation takes over exploration in the proposed EM  
algorithm, representing by the convergence of the control covariance matrix. }

Let us first define the following notation
\begin{align} \label{eq_cal_I_def}
 \mathcal{I}  (\hat{\phi}^i, \hat{\phi}^{i+1}) =
  \mathbf{I} -  \left({  \nabla^2_{ \phi} \mathcal{L}(\phi, \hat{\phi}^i) \big|_{ \phi=  \hat{\phi}^{i+1}} }\right)^{-1}   \nabla^2_{ \phi}  L_\phi (\mathbb{Y}_T) \big|_{\phi=\hat{\phi}^i}       
\end{align}  
that can be used to describe the convergence error from $\hat{\phi}^{i} - \hat{\phi}_{EM}$
to  $\hat{\phi}^{i+1} - \hat{\phi}_{EM}$ in the following theorem.

\begin{theorem}  Let 
$\hat{\phi^i}$ be a  known parameter estimate and $\hat{\phi}^{i+1}$ satisfy
\begin{align} \label{parcalLzero}
\frac{\partial \mathcal{L}(\phi, \hat{\phi}^i)}{\partial \phi} \Big|_{\phi= \hat{\phi}^{i+1}} =\mathbf{0}.
\end{align}
 Then, for $\hat{\phi}_{EM}$ defined  in Lemma~\ref{lemma-limit},
 \begin{align} \label{eq:169}
    \hat{\phi}^{i+1} - \hat{\phi}_{EM} =    \mathcal{I}  (\hat{\phi}^i, \hat{\phi}^{i+1}) (\hat{\phi}^i-\hat{\phi}_{EM} )
    +o (\hat{\phi}^i-\hat{\phi}_{EM} )
\end{align}
where the notation $o$ represents higher order smallness. 
\end{theorem}

\begin{proof}
%
One can utilize the Taylor series expansion of $\frac{\partial L_\phi(\mathbb{Y}_T) }{\partial \phi}$ about $\phi=\hat{\phi}^i$ and evaluate it at $\phi = \hat{\phi}_{EM}$ as follows
\begin{align*}  
0= & \frac{\partial L_{\phi} (\mathbb{Y}_T) }{\partial \phi} \Big|_{\phi= \hat{\phi}_{EM}}  =
\frac{\partial L_{\phi} (\mathbb{Y}_T) }{\partial \phi} \Big|_{\phi= \hat{\phi}^i} \\
& + 
\nabla^2_{ \phi}  L_\phi (\mathbb{Y}_T) \big|_{\phi=\hat{\phi}^i}   (\hat{\phi}_{EM}-\hat{\phi}^i )    +o (\hat{\phi}^i-\hat{\phi}_{EM} ).
\end{align*}
The first equation holds because $\hat{\phi}_{EM}$  is a stationary point of $L_\phi(\mathbb{Y}_T)$ by Lemma~\ref{lemma-limit}.
As a result, 
	\begin{align}\label{eq174}  
 \nabla^2_{ \phi}  L_\phi (\mathbb{Y}_T) \big|_{\phi=\hat{\phi}^i}   (\hat{\phi}^i -\hat{\phi}_{EM})  = \frac{\partial L_{\phi} (\mathbb{Y}_T) }{\partial \phi} \Big|_{\phi= \hat{\phi}^i}
  +o (\hat{\phi}^i-\hat{\phi}_{EM} ).
\end{align}

Again one can apply the Taylor series expansion of $\frac{\partial \mathcal{L}(\phi, \hat{\phi}^i) }{\partial \phi}$ about $\phi=\hat{\phi}^i$ and evaluate it at $\phi=\hat{\phi}^{i+1}$ as follows, noting the explicit quadratic expression of $\mathcal{L}$,
\begin{align*} 
    \frac{\partial \mathcal{L}({\phi},\hat{\phi}^i)}{\partial \phi}\big|_{\phi= \hat{\phi}^{i}}  = & \frac{\partial \mathcal{L}(\phi, \hat{\phi}^i)}{\partial \phi} \big|_{\phi= \hat{\phi}^{i+1}} \\& +  {\nabla^2_{\phi}  \mathcal{L}(\phi, \hat{\phi}^i) \big|_{ \phi=  \hat{\phi}^{i+1}}  } (\hat{\phi}^i - \hat{\phi}^{i+1}) 
 \end{align*}
which implies, due to \eqref{parcalLzero},
 \begin{align} \label{eq176}
    \frac{\partial \mathcal{L}({\phi},\hat{\phi}^i)}{\partial \phi}\big|_{\phi= \hat{\phi}^{i}}  =    {\nabla^2_{\phi}  \mathcal{L}(\phi, \hat{\phi}^i) \big|_{ \phi=  \hat{\phi}^{i+1}}  } (\hat{\phi}^i - \hat{\phi}^{i+1}).
 \end{align}
 
Finally, by \eqref{JacobianLcalL} of Lemma~\ref{lemma-LcalL-Jac}, one can equate  \eqref{eq174} and \eqref{eq176} as
\begin{align*}
\nabla^2_{ \phi}  L_\phi (\mathbb{Y}_T) \big|_{\phi=\hat{\phi}^i}   (\hat{\phi}^i -\hat{\phi}_{EM})  =&
 {\nabla^2_{\phi}  \mathcal{L}(\phi, \hat{\phi}^i) \big|_{ \phi=  \hat{\phi}^{i+1}}  } (\hat{\phi}^i - \hat{\phi}^{i+1})  \\ &+o (\hat{\phi}^i-\hat{\phi}_{EM} ).
\end{align*}
and hence \eqref{eq:169}. \end{proof}

Next, we discuss the covariance matrix component $\hat\mbsigma_k^i$ of $\hat\phi_k^i =\col(\hat\mbf_k^i, \hat\mbe_k^i, \hat\mbsigma_k^i)$.  
Denote $\hat\mbsigma^i = \col( \hat\mbsigma_1^i, \cdots,\hat\mbsigma_T^i)$.  
For $\hat{ \phi}^{i+1} = \col(\hat{ \phi}^{i*}_1, \cdots, \hat{ \phi}^{i*}_T)$ calculated using the
   \textbf{SOC-EM II} algorithm \eqref{optim2}  [similar analysis holds for the \textbf{SOC-EM I} algorithm \eqref{optiphij}],  by Theorem~\ref{uniqueness1}, $\hat{ \phi}^{i*}_j$ is the solution to  \eqref{jacobian_set_zero}, that is, 
 \begin{align} \frac{ \partial \sum_{k=1}^T  \Bar{\mathcal{L}}_k(\phi_j,\hat\phi^i)  }{\partial \phi_j}  \Big|_{\phi_j= \hat{\phi}_j^{i*}}  = \mathbf{0}, \; j=1,\cdots, T.
\end{align}
It approximately implies that \eqref{parcalLzero} is satisfied, and hence
  \begin{align} \label{hatphiEM1}
    \hat{\phi}^{i+1} - \hat{\phi}_{EM} =    \mathcal{I}  (\hat{\phi}^i, \hat{\phi}^{i+1}) (\hat{\phi}^i-\hat{\phi}_{EM} )
   \end{align}
holds with the higher order smallness ignored.

By Lemma~\ref{lemma-limit}, one has $\lim_{i \rightarrow \infty} \hat{\phi}^{i} = \hat{\phi}_{EM}$  
if $\hat{\phi}^{i}$   recursively generated by $\hat\phi^{i+1} = \hat{ \phi}^{i*}$ according to \eqref{maxlmaxcalL},
approximated by  \textbf{SOC-EM II} \eqref{optim2}. It is noted that 
$\hat\mbsigma_k^i$ and ${\hat\mbsigma}_{k,EM}$
are the covariance matrix component of $\hat{\phi}^{i}$  and $\hat{\phi}_{EM}$, respectively. 
As shown in Remark~\ref{remark-zerocovariance},  one has  $\hat\mbsigma_k^{i*}=\mathbf{0}$
and hence ${\hat\mbsigma}_{k,EM}=\mathbf{0}$. 
Now, from \eqref{hatphiEM1}, one has approximately, 
  \begin{align} \label{Sigmareduction}
\hat\mbsigma^{i+1}  =    \mathcal{I}_\Sigma  (\hat{\phi}^i, \hat{\phi}^{i+1}) \hat\mbsigma^{i}
   \end{align}
 for some  $\mathcal{I}_\Sigma$. 
 Based on Theorem~\ref{theorem:global_maximizer} and its application to \eqref{eq_cal_I_def}, one can approximately conclude that 
  \begin{align} \label{eq:diag_I_major}
&  \mathcal{I}  (\hat{\phi}^i, \hat{\phi}^{i+1}) = \text{diag}
  \begin{bmatrix} 
 \mathcal{I} _1 (\hat{\phi}^i, \hat{\phi}_1^{i+1})     \\
    \vdots  \\
 \mathcal{I}_T  (\hat{\phi}^i, \hat{\phi}_T^{i+1})    
    \end{bmatrix} 
    \leq \mathbf{I} 
    \end{align}
where the information matrix $ \mathcal{I} _j (\hat{\phi}^i, \hat{\phi}_j^{i+1})  $ corresponds to the counterpart of the component of  $  - \nabla^2_{\phi_j} \{ \sum_{k=1}^T  \bar{\mathcal{L}} (\phi_j,\hat\phi_k^i) \} |_{\phi_j = \hat{\phi}_j^{i+1} }$. Furthermore \textbf{SOC-EM II} is carried out for all time instants separately with no correlation between them. Therefore one can stack them in a matrix with diagonal elements as the individual (for each optimization instant) Hessians to create a higher dimensional Hessian which essentially provides property of the policy.
  Similarly the principal minor of the $ \mathcal{I}  (\hat{\phi}^i, \hat{\phi}^{i+1})$  concerned with the covariance components inherits the property from \eqref{eq:diag_I_major} and follows the following inequality,
    \begin{align} \label{bounds_of_information}
\mathbf{0}    \leq \mathcal{I} _\Sigma (\hat{\phi}^i, \hat{\phi}^{i+1}) \leq \mathbf{I}.
 \end{align}

The equation \eqref{Sigmareduction} is trivially true because  
 $ \hat\mbsigma_k^{i+1}  = \hat\mbsigma_k^{i*}=\mathbf{0}$ recursively  in   \textbf{SOC-EM II}. 
 However, in real scenarios, \textbf{SOC-EM II} cannot be perfectly implemented, but practically in the sense of 
   \begin{align}
  \| \hat{ \phi}^{i+1}  - \col(\hat{ \phi}^{i*}_1, \cdots, \hat{ \phi}^{i*}_T) \| <\Delta,
      \end{align}
for some error tolerance $\Delta$. As a result,  
 $\hat\mbsigma_k^{i+1}  = \mathbf{0}$ does not hold anymore. 
 Nevertheless, \eqref{Sigmareduction} can approximately claim that
 $\hat\mbsigma^{i}$ converges to zero
 as $i$ goes to $\infty$. 
 In particular, the following theorem states the conclusion  in terms of 
 the singular values of the covariance matrices
  ${\hat\mbSigma_k}^i  =({\hat\mbSigma_k} ^{\frac{1}{2} i}) ^\top {\hat\mbSigma_k} ^{\frac{1}{2} i} $ under the condition 
\eqref{bounds_of_information}, where 
$\hat\mbsigma_k^i =  \text{vec} (\hat\mbSigma_k^{\frac{1}{2}i})$.

%
%
%
%

\begin{theorem}
Suppose \eqref{Sigmareduction} holds with \eqref{bounds_of_information}. 
Let $\sigma_{k,1}^\imath,\cdots, \sigma_{k,n_a}^\imath$  be the singular values of 
$ \mbSigma_k^\imath$  for $\imath =i, i+1$,   then
\begin{align} \label{singularvalue}
     \sum_{k=1}^T  \sum_{p=1}^{n_a}  \sigma_{k,p}^{i+1} \leq \sum_{k=1}^T \sum_{p=1}^{n_a} \sigma_{k,p}^{i}.
\end{align}
\end{theorem}

\begin{proof} The equation  \eqref{Sigmareduction} multiplied by its transpose gives
  \begin{align*}  
 (\hat\mbsigma^{i+1})^\top 
  \hat\mbsigma^{i+1}   =& (\hat\mbsigma^{i})^\top 
     \mathcal{I}_\Sigma  (\hat{\phi}^i, \hat{\phi}^{i+1})^\top    \mathcal{I}_\Sigma  (\hat{\phi}^i, \hat{\phi}^{i+1})  \hat\mbsigma^{i} \\
     \leq &
      (\hat\mbsigma^{i})^\top  \hat\mbsigma^{i},
   \end{align*}
 where the inequality holds due to \eqref{bounds_of_information}. It is equivalent to
   \begin{align*}  
  \sum_{k=1}^T  (\hat\mbsigma_k^{i+1})^\top \hat\mbsigma_k^{i+1}
    \leq  
     \sum_{k=1}^T  (\hat\mbsigma_k^{i})^\top  \hat\mbsigma_k^{i}.   \end{align*}
   Utilizing the following property, for  $\imath =i, i+1$, 
   \begin{align*}  (\hat\mbsigma_k^{\imath})^\top  \hat\mbsigma_k^{\imath} &= \Tr  (({\hat\mbSigma_k} ^{\frac{1}{2} \imath}) ^\top {\hat\mbSigma_k} ^{\frac{1}{2} \imath} )  = \Tr   ({\hat\mbSigma_k}^\imath)      
    =\sum_{p=1}^{n_a}  \sigma_{k,p}^\imath \end{align*} yields \eqref{singularvalue}.
\end{proof}

\section{Experimental Results} \label{sec:results}
In this section, we investigate empirical performance of
the proposed  SOC-EM algorithm that aims to  utilize the parameter estimates of baselines as mentioned in Section~\ref{sec:practical} to deliver better controller parameters. 
The experiments were conducted  on a \BoxD\;framework that is a rigid polygon mass subject to gravity, linear and angular damping  \cite{parberry2013introduction}. 
 We consider sensor noise $\mbepsilon^s_k$ to be Gaussian in our experiments, i.e.,
\begin{align}
\mathbf{s}_k &= \mathbf{x}_k+ \mbepsilon^s_k ,\;   \mbepsilon^s_k \sim \mathcal{N} (0,  \rho^2 \mathbf{I}_{n_s} ), 
\label{cov_a_k}
\end{align}
where $\mathbf{x}_k$ is the real state. 
The sensor noise $\mbepsilon^s_k$  is propagated into the design of control action $\mathbf{a}_k$ that  forms the real  input 
to the system.  The states of the system are $\mathbf{x}= \{x,y,v_x,v_y\}^\top$ representing the position and velocity in a 2D environment.  
 With the control action $\mathbf{a}_k = [a_x, a_y] \in \mathbb{R}^2$, the objective is driven from $[0,5]$ to $[5,20]$
 in the Cartesian coordinate and stay there using the shortest time. 
 Algorithm~\ref{alg:EM} was implemented in the experiments.

\begin{figure}[t]
	\centering
	\includegraphics[scale=0.46]{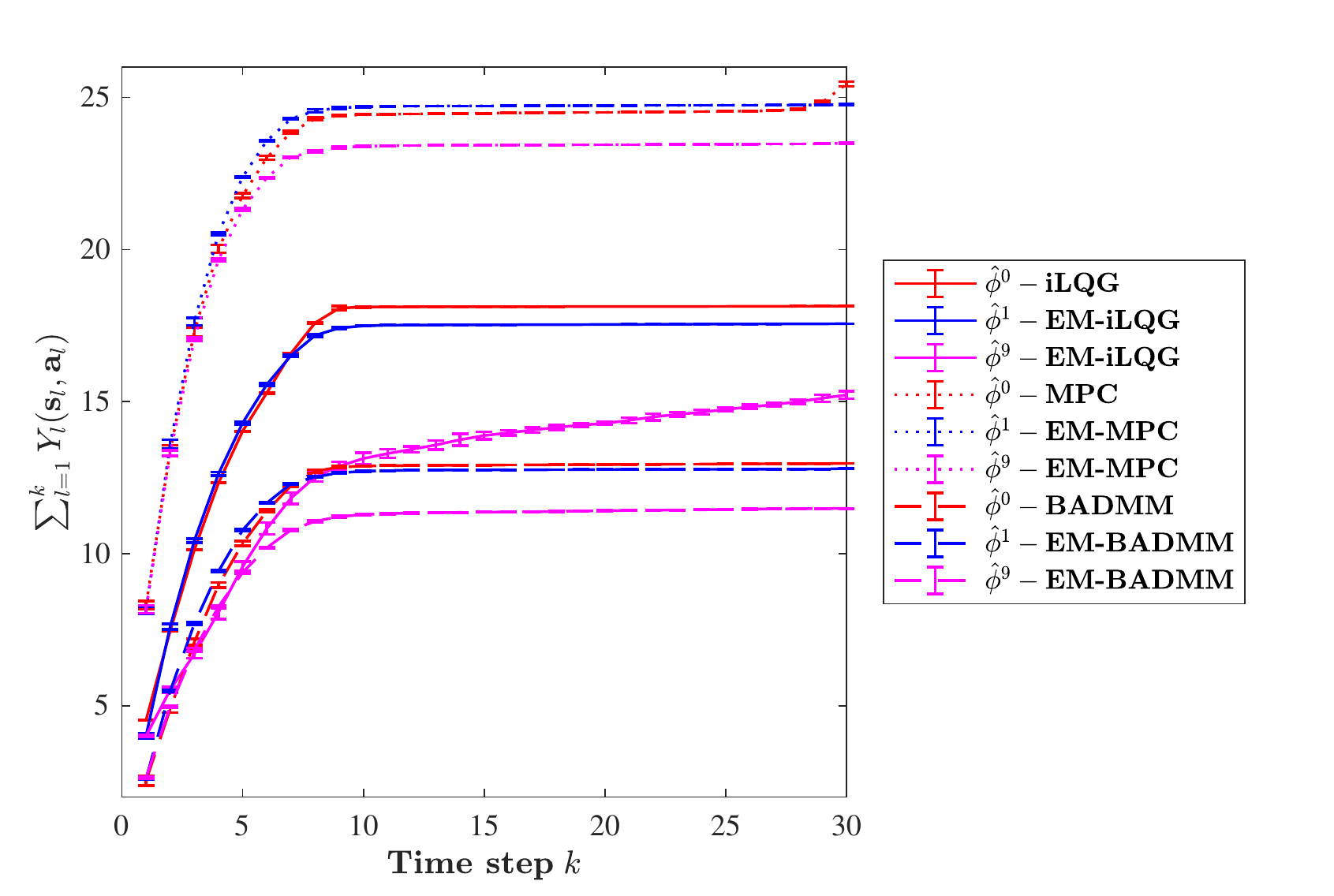}
	\caption{Cumulative sum of costs evaluated for three baselines, 1) iLQG, 2) MPC, and 3) BADMM,  for  
		different controller parameters $\hat{\phi}^0, \hat{\phi}^1$ and $\hat{\phi}^9$.}
	\label{fig:reward_3_cases}
\end{figure}

\begin{figure}[t]
\centering
	\includegraphics[scale=0.265]{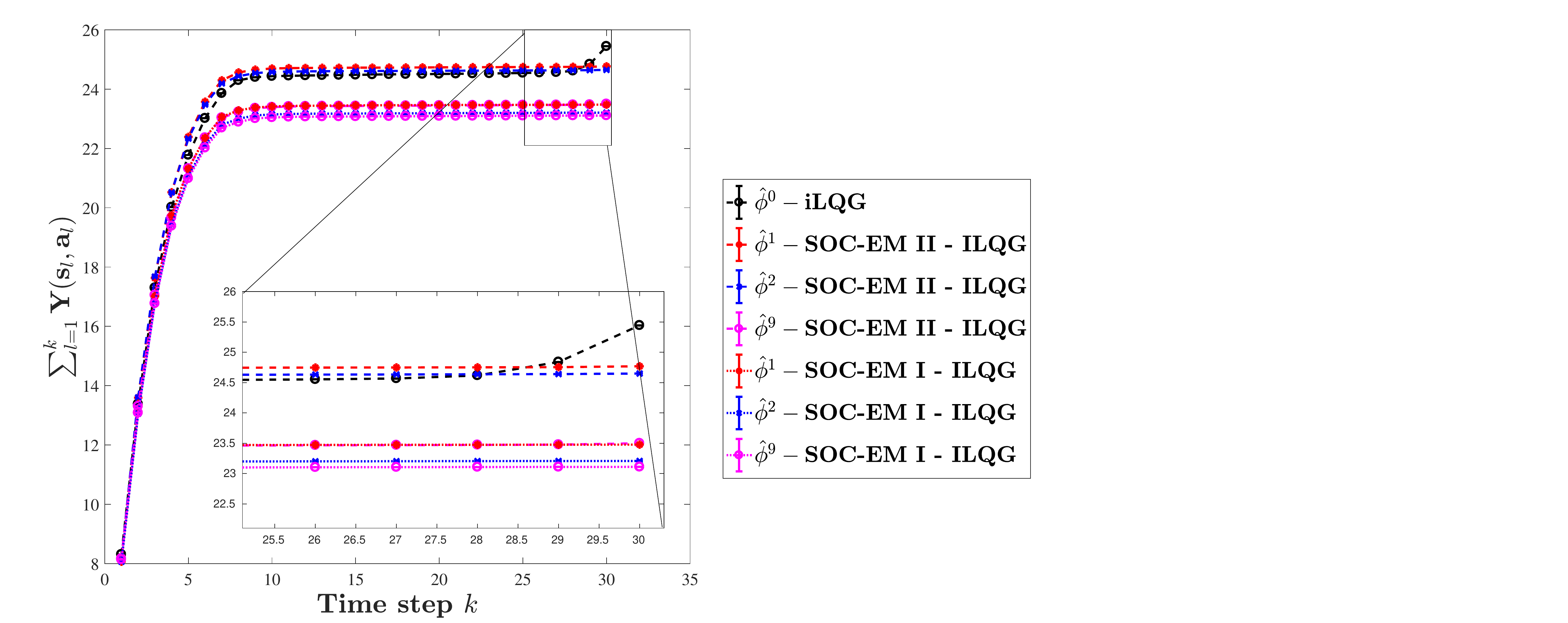}
	\caption{Comparison of SOC-EM 1 and SOC-EM II for three specific cases each corresponding to iLQG baseline.}
	\label{fig:soc-em1/2}
\end{figure}

{\it Cost:} The mean$\pm$std-dev of the cumulative sum of the real  costs
$\sum_{l=1}^k Y_l(\mathbf{s}_l,\mathbf{a}_l), \; k=1,\cdots, T=30$, is depicted in Fig.~\ref{fig:reward_3_cases}. The datasets take into account $20$ samples of cumulative sum of the costs against the time steps evaluated on three baselines.
The noise parameter was set as $\rho=0.3$. The plots in red, blue and magenta represent the performance with the
 controllers parameterized by $\hat{\phi}^0$, $\hat{\phi}^1$, and $\hat{\phi}^{9}$, respectively. 
The  solid, dotted, and dashed plots delineate three baselines iLQG,  MPC,  and BADMM, respectively. 
For the three baselines, it is observed that 
the mean$\pm$std-dev of the cumulative sum decreases over subsequent iterations, which verifies that 
 the {cost-to-go} is reduced through iterative EM procedures.    
 {Figure~\ref{fig:soc-em1/2} exhibits the cost-to-go comparison of the SOC-EM I and SOC-EM II algorithms
 with the iLQG baseline for three specific controller parameters.  It is observed that 
SOC-EM I works better than SOC-EM II because the former is a more accurate high dimensional optimization
as compared to the latter.  
}

\begin{figure}[t]
\centering
\includegraphics[width=.5\textwidth]{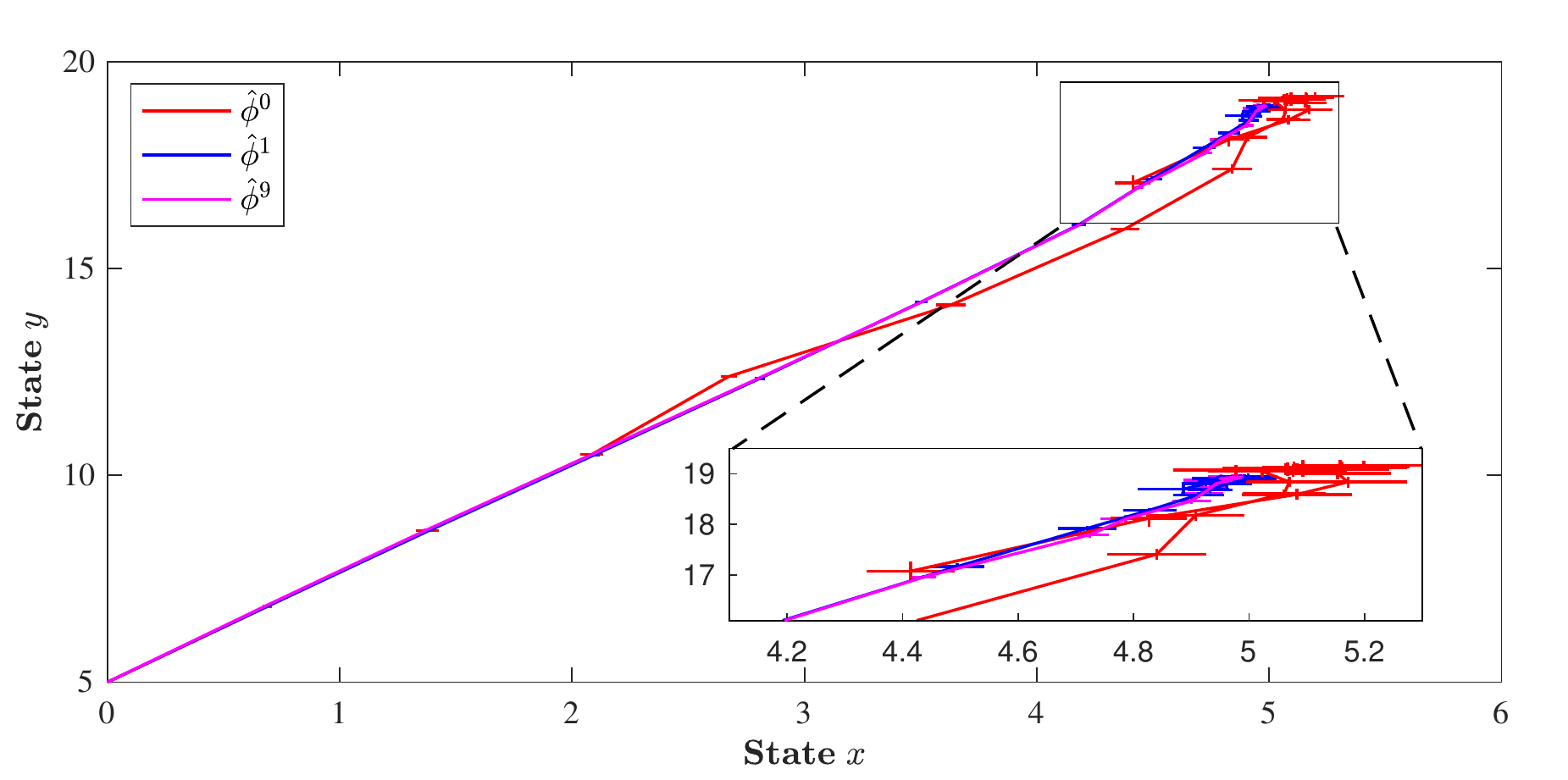}
\caption{Profile of the position trajectories  with different controller parameters.}
\label{fig:trajectories}
\medskip
%
\centering
\includegraphics[width=.5\textwidth]{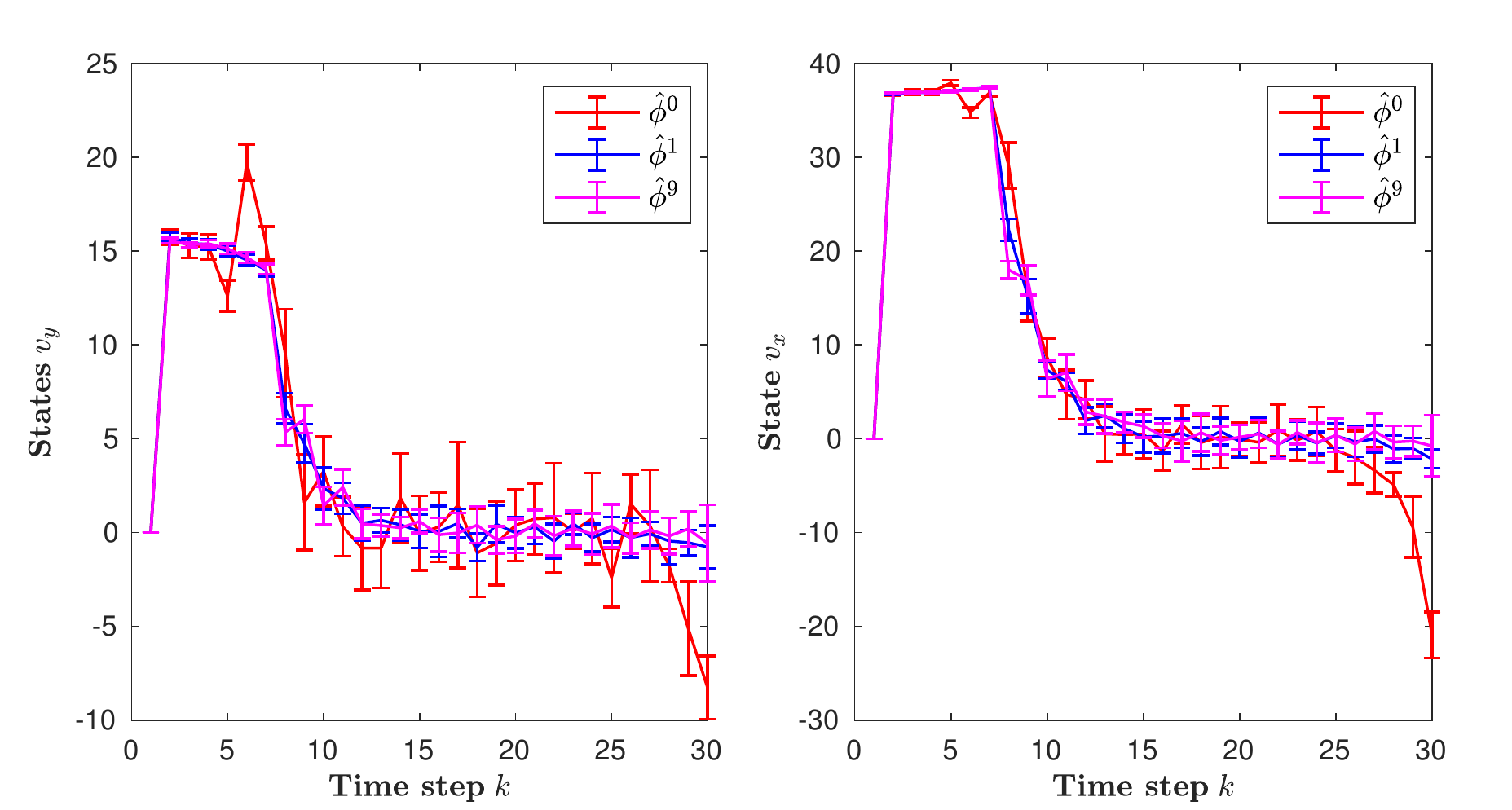}
\caption{Profile of the velocity trajectories  with different controller parameters.}
\label{fig:volocity}
\end{figure}

 {\it Trajectories:} We compare the true state trajectories, i.e., $ \{\textbf{x}_1,..,\textbf{x}_T\}$ produced on the real platform excited 
 by the control actions with the parameters obtained through repeated EM iterations. In this experiment, iLQG was used 
 as the baseline. We performed 30 experiments and each experiment ran for 10 subsequent iterations.
 The mean$\pm$std-dev of the trajectories $[x_k, y_k]$, $k=1,\cdots,30$, is illustrated in Fig.~\ref{fig:trajectories}
 for the control action with $\hat{\phi}^0$, $\hat{\phi}^1$, and $\hat{\phi}^{9}$. 
 It  can also be observed that all the trajectories move to the proximity of the target $[5, 20]$ quickly (in approximately 8 steps)
 and stay there in the remaining steps. 
The magenta trajectory for $\hat{\phi}^9$ of less jittery nature in contrast to the red one for $\hat{\phi}^0$
demonstrates better performance achieved by EM iterations.
In particular, there is a high variance in the red trajectory  near the final time step as the influence of noise is accumulated temporally. 
The corresponding velocity trajectories $v_x$ and $v_y$  versus time are illustrated in Fig.~\ref{fig:volocity}.
It is observed that the velocities increase to the maximum to drive the object to the target position quickly (again in approximately 8 steps) and then decrease to zero. The advantage gained by the EM iterations can be explained by the less deviation caused by noise  in the magenta trajectory.

  \begin{figure}[t]
\centering
\hspace{-3mm}
\includegraphics[width=.5\textwidth]{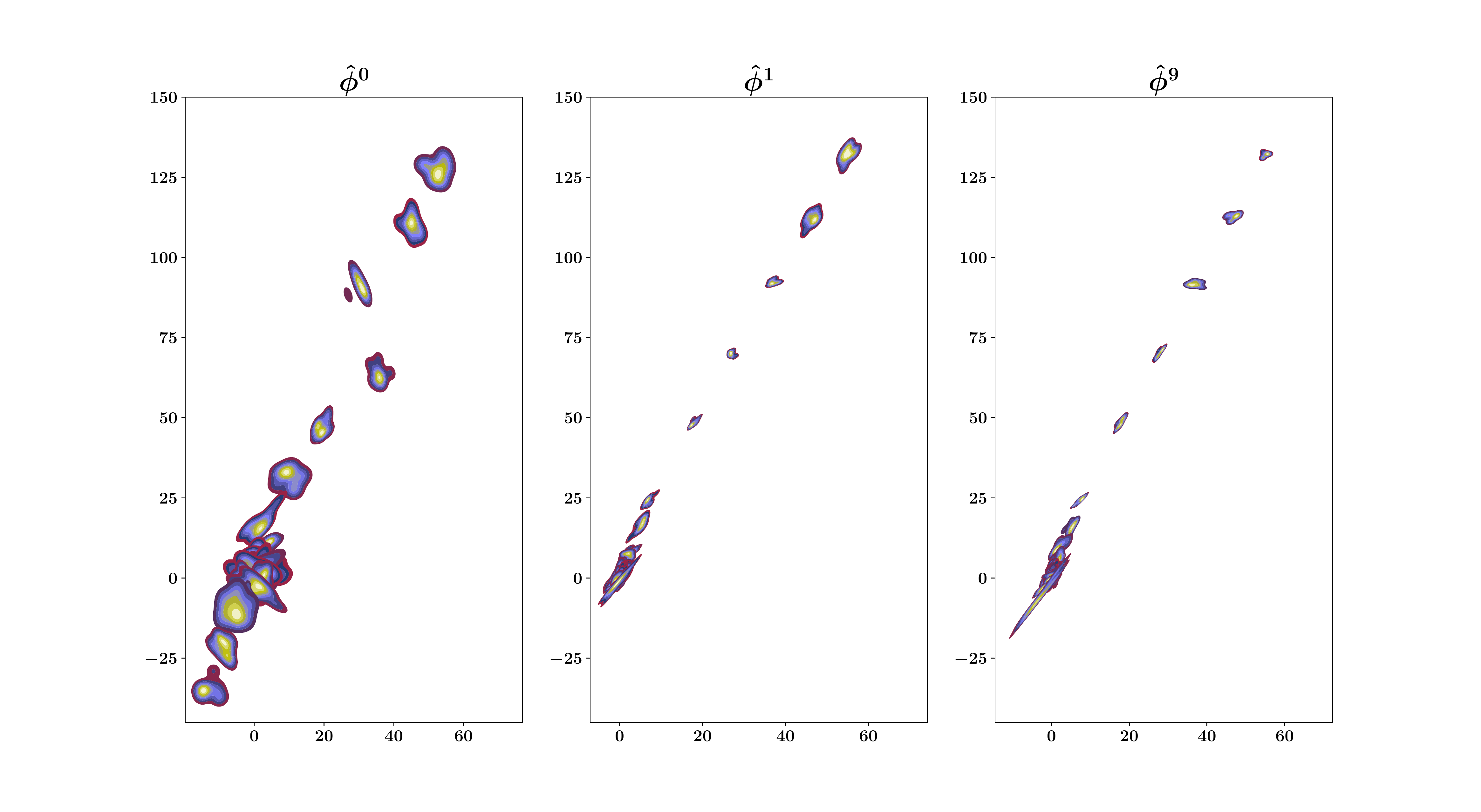}
\caption{Plot of 2D kernel density of control actions for 100 samples.}
\label{fig:KDE_Plot}
\medskip
\centering
\includegraphics[width=.52\textwidth]{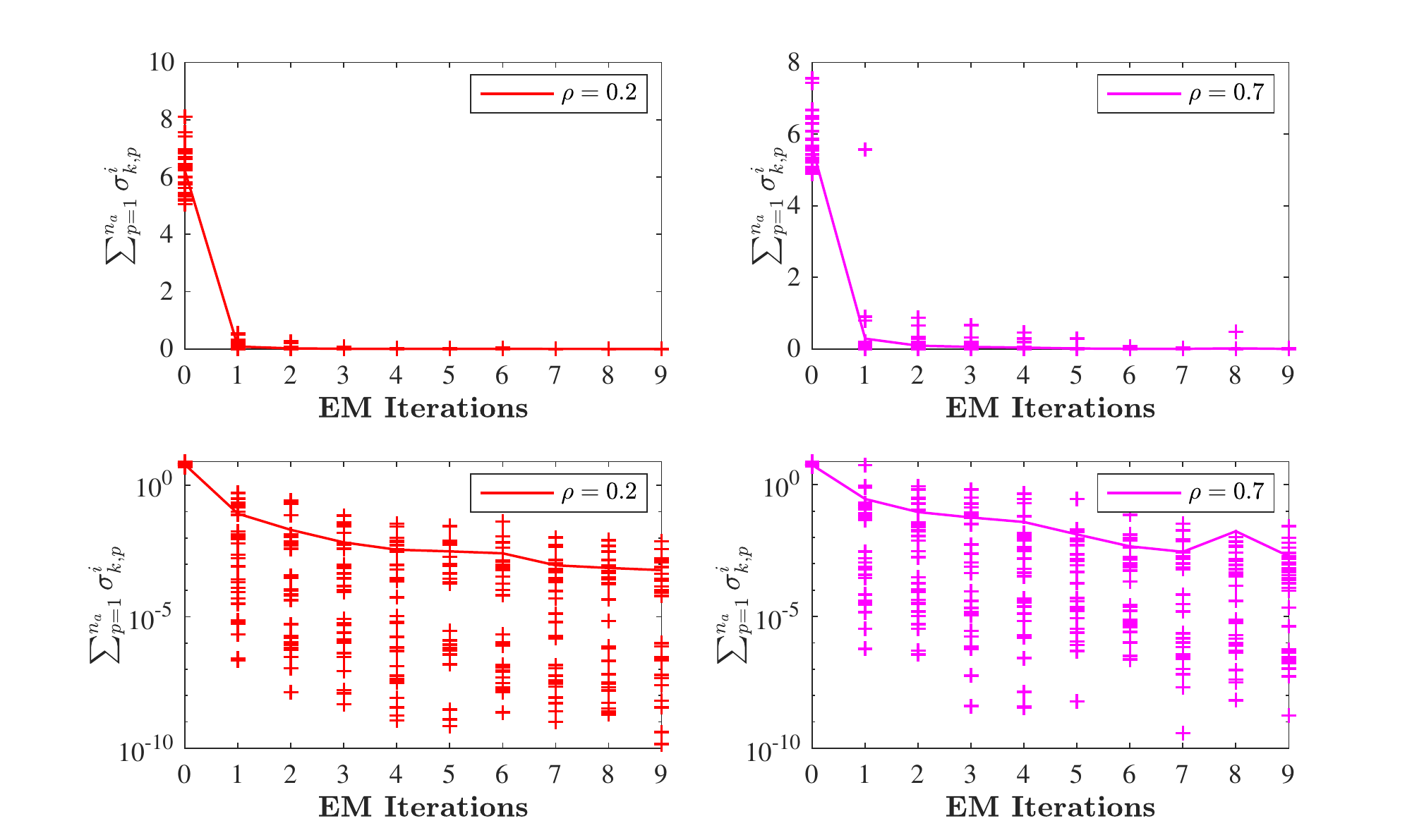}
\caption{Sum of singular values of the covariance matrix 
vs. EM iterations for $\rho=0.2$ and $0.7$ in linear and log scales.}
\label{fig:singular_Values_ltv}
\end{figure}

{\it Control actions: }  Figure~\ref{fig:KDE_Plot} shows the evolution of control actions in terms kernel density plots of samples collected
from 100 experiments. 
The control actions should be generated to maximum for a large velocity at the start and then reduced to zero in an ideal environment. 
It is evident that  the control actions with the parameters $\hat{\phi}^0$ 
do not well settle down to zero. 
On the contrary the control actions as a result of $\hat{\phi}^{1}$ and $\hat{\phi}^9$ are of
significant improvement. 


\begin{figure}[t]
	\centering
	\includegraphics[width=.52\textwidth]{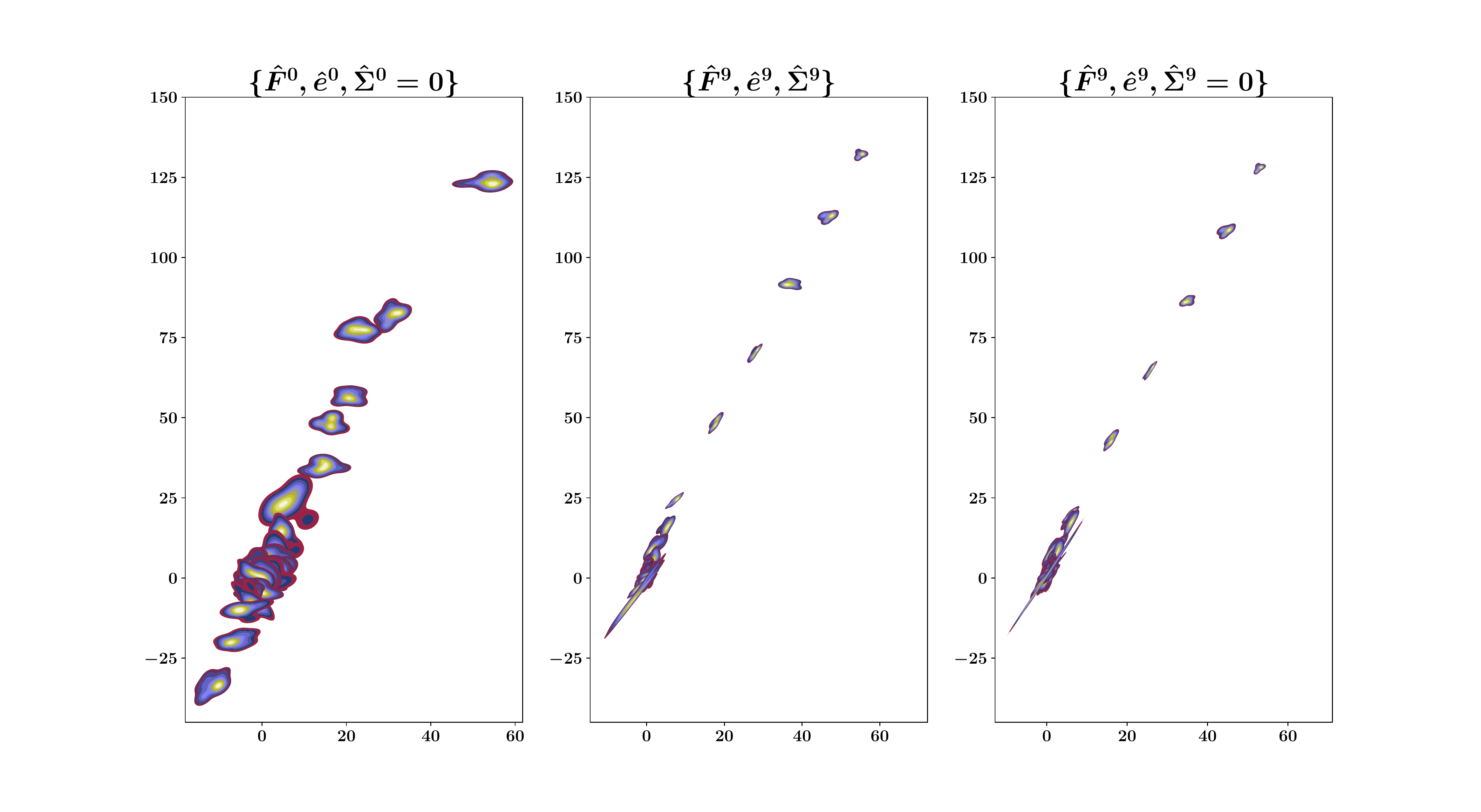}
	\caption{Plot of 2D kernel density of control actions for three policy parameters a) converged iLQG parameters with $\hat{\mbSigma}^0=\pmb{0}$, b) EM-iLQG parameters i.e., $\hat{\phi}^9$ and c) EM-iLQG parameters $\hat{\phi}^9$ with $\hat{\mbSigma}^9 = \pmb{0}$. }
	\label{fig:new_KDE_Plot_zero_noise}
	\medskip
\centering
\includegraphics[width=.52\textwidth]{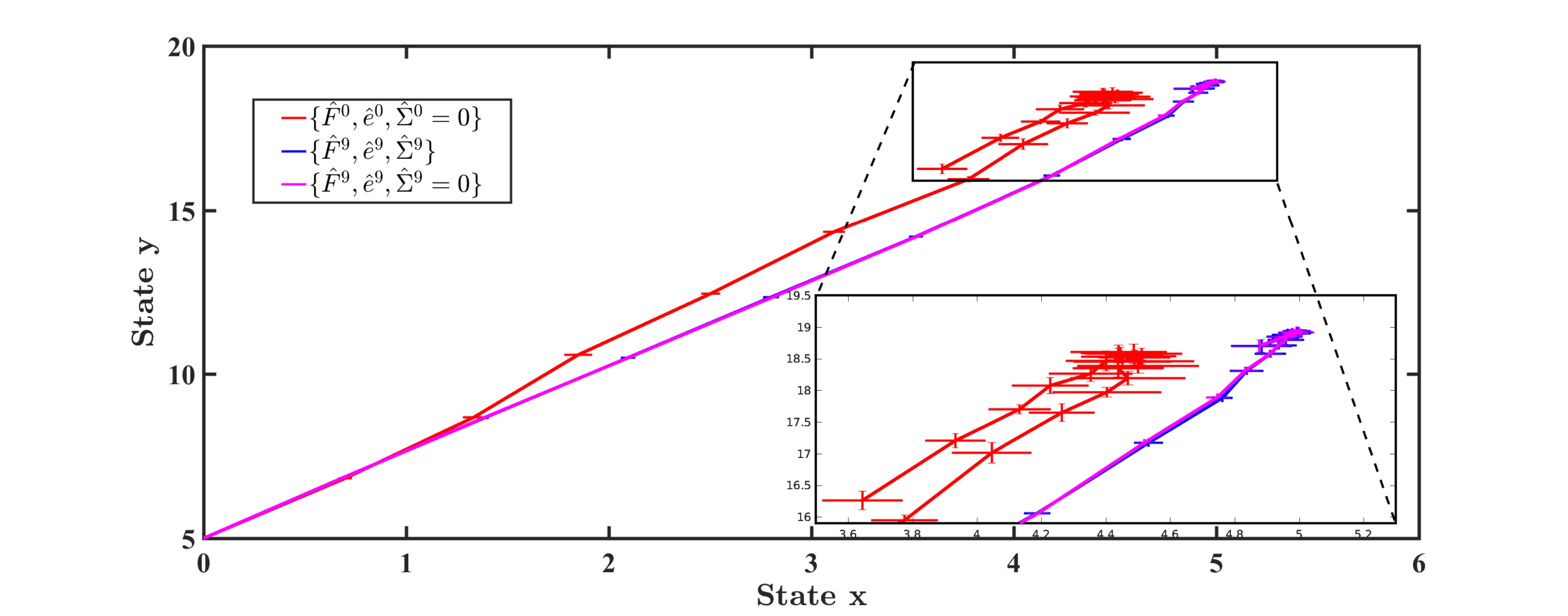}
\caption{Plot of true state trajectory for three policy parameters a) converged iLQG parameters with $\hat{\mbSigma}^0=\pmb{0}$, b) EM-iLQG parameters i.e., $\hat{\phi}^9$ and c) EM-iLQG parameters $\hat{\phi}^9$ with $\hat{\mbSigma}^9 = \pmb{0}$. }
\label{fig:new_traj_zero_noise_Plot}
\end{figure}
 
{{\it Exploitation efficiency:} The observed improvement in control actions 
can be well manifested by the exploitation mechanism in the EM approach
which significantly reduces  the stochasticity in the control policies.}
In particular,  the theoretical analysis in control covariance matrix in Section~\ref{section:noise}
can be verified by the plots of iterative decrease 
in the sum of the singular values of covariance matrices; see Fig.~\ref{fig:singular_Values_ltv}.
We simulated the entire procedure of EM with different noise factors and recorded  
$\sum_{p=1}^{n_a} \sigma_{k,p}^{i}$ which is marked as `+` for each time step $k=1,\cdots,30$ and
each EM iteration $i=0,\cdots, 9$. The average (equivalent to the sum divided by 30) is represented 
by the solid curve.  
The plots in log scale better shows the decrease pattern for  $\rho=0.2$ as expected by the theory. 
It can be noted that for $\rho=0.7$, the pattern is violated at $i=8$, which is due to the higher order smallness
in \eqref{eq:169}.

{\it Measurement noise:} 
{It has been exhibited that the exploitation functionality of the EM approach 
can effectively reduce the stochasticity in control policies. Next, we will further highlight 
this effectiveness in comparison with artificially setting the covariance matrices zero.
The comparison is made among the three cases, 
i) the iLQG parameter $\hat\phi^0$ with $\hat{\mbSigma}^0=\pmb{0}$, ii) 
the EM-iLQG parameter $\hat{\phi}^9$,  and iii) $\hat{\phi}^9$ with $\hat{\mbSigma}^9 = \pmb{0}$. 
Indeed, it is observed that artificially setting the covariance matrices zero
does not satisfactorily reduce the stochasticity in the control policies, 
as the stochasticity propagated from measurement noise is unavoidable. 
The kernel density plots Fig.~\ref{fig:new_KDE_Plot_zero_noise} shows that 
the EM approach performs better than the iLQG with zero covariance matrices.
It is not surprising to see that the difference between the cases (ii) and (iii) is minor as 
$\hat{\phi}^9$ is closely approaching 0 using SOC-EM.
The corresponding state trajectories plotted in Fig.~\ref{fig:new_traj_zero_noise_Plot} support the 
same conclusion.}

The simulations of Step 1, 2 and 4 were performed using the 64-bit \textsf{Ubuntu} 16.04 OS on  \textsf{Dell Alienware} 15 R2 
of  \textsf{Intel} Core i7-6700HQ CPU @ 2.60GHz.
The simulations of Step 3 were conducted using multiple 2.6 GHz \textsf{Intel Xeon Broadwell (E5-2697A v4)} processors on the high performance computing (HPC) grid located at The University of Newcastle. We switched processors in order to leverage parallel processing of the optimization routine of \eqref{optim2}.



\section{Conclusions} \label{sec:conclusion}

This paper has proposed a new EM based methodology for solving the SOC problem, resulting in 
an SOC-EM algorithm. The method effectively bridges the relationship between the optimal control problem and the EM algorithm that is originally used for maximizing the likelihood of observed data. Moreover, we have discussed 
a practical solution to SOC-EM and the uniqueness of controller parameter estimation.
The algorithm has been applied to the \BoxD\;framework and
the experiments support the superiority of the new technique, compared to some of widely known and extensively employed methodologies.
The paper has established a new research framework that has potential development in the future work. 
For example, nonlinear stochastic dynamics,  persistently exciting property
of a system as a result of parameters obtained through EM, input constraints,  fitting stable linear dynamic models, etc., 
are  interesting topics.

\appendix    

 \subsection{Derivation of the terms in Lemma~\ref{lemma:LTV}} \label{appednixB}
 The terms of $\Theta_1(\phi_k), \Theta_2(\phi_k)$, and $\Theta_3(\phi_k)$ after expansion are shown below.
 First, 
\begin{align*} 
    \Theta_1(\phi_k) &=    \mathbb{E}_{\hat{\phi}^i } ({\mbzeta_k \mbzeta_k ^\top | \mathbb{Y}_{T}) }= \begin{bmatrix}
        \mathbf{G}_{k+1} & \gamma_1 \\
      \gamma_1^\top & \gamma_2  \\
      \end{bmatrix}
\end{align*}
where 
\begin{align*}
    {\gamma_1^\top} =& \mathbf{A}^r_k \mathbb{E}_{\hat{\phi}^i} (\mathbf{s}_k \mathbf{s}_{k+1}^\top|\mathbb{Y}_{T}) + \mathbf{B}^r_k \mathbb{E}_{\hat{\phi}^i }  (\mathbf{a}_k \mathbf{s}_{k+1}^\top | \mathbb{Y}_{T}) \\
    \gamma_2=& \mathbf{A}^r_k \mathbf{G}_k  {\mathbf{A}^r_k}^\top + \mathbf{B}^r_k \mathbb{E}_{\hat{\phi}^i} ( {\mathbf{a}_k \mathbf{a}_k^\top} |\mathbb{Y}_T )  {\mathbf{B}^r_k}^\top \\
    &+ \mathbf{A}^r_k  \mathbb{E}_{  \hat{\phi}^i} (\mathbf{s}_k \mathbf{a}_k^\top |\mathbb{Y}_T) {\mathbf{B}^r_k}^\top + (\mathbf{A}^r_k \mathbb{E}_{\hat{\phi}^i} (\mathbf{s}_k \mathbf{a}_k^\top |\mathbb{Y}_T) {\mathbf{B}^r_k}^\top)^\top +   \mbSigma^r_k.
\end{align*}
Similarly, the matrix $\Theta_2 (\phi_k)$ can be expanded as
\begin{align*} 
    \Theta_2(\phi_k) &=    \mathbb{E}_{\hat\phi^i} ({\mbzeta_k \mbz_k ^\top | \mathbb{Y}_{T}) }=   \begin{bmatrix}
        \mathbf{M}_{k+1|T} & \gamma_3  \\
       \gamma_4 & \gamma_5  \\
      \end{bmatrix}
\end{align*}
where 
\begin{align*}
    \gamma_3 &= \mathbf{M}_{k+1|T} \mbF_k^\top + \hat{\mathbf{s}}_{k+1|T} \mbe_k^\top  \\
    \gamma_4 &= \mathbf{A}^r_k \mathbf{G}_k + \mathbf{B}^r_k \mathbb{E}_{\hat{\phi}^i} (\mathbf{a}_k \mathbf{s}_k^\top | \mathbb{Y}_{T})\\
    \gamma_5 &= \mathbf{A}^r_k \mathbf{G}_k \mbF_k^\top + \mathbf{A}^r_k \hat{\mathbf{s}}_{k|T} \mbe_k^\top + \mathbf{B}^r_k \mathbb{E}_{\hat{\phi}^i} (\mathbf{a}_k \mathbf{a}_k^\top | \mathbb{Y}_{T}).
\end{align*}

The matrix $\mathbf{A}^o_k \Theta_3(\phi_k)  {\mathbf{A}^o_k}^\top$ has the expression 
\begin{align*} 
 & \mathbf{A}^o_k \Theta_3(\phi_k)  {\mathbf{A}^o_k}^\top   = \begin{bmatrix}
\gamma_8 & \gamma_9  \\
\gamma_{10}  & \gamma_{11}\\
\end{bmatrix},
\end{align*}
where
\begin{align*}
\gamma_6 &=  \mathbf{G}_k \mbF_k^\top + \hat{\mathbf{s}}_{k |T}\mbe_k^\top  \\
 \gamma_7 &= \mathbb{E}_{\hat{\phi}^i} (\mathbf{a}_k \mathbf{a}_k^\top | \mathbb{Y}_{T} )  \\
 \gamma_8  &= {\mathbf{A}^d_k} (\mathbf{G}_k {\mathbf{A}^d_k}^\top + \gamma_6 {\mathbf{B}^d_k}^\top  ) + 
{\mathbf{B}^d_k} ( {\gamma_6}^\top {\mathbf{A}^d_k}^\top +  \gamma_7 {\mathbf{B}^d_k}^\top  ) \\
\gamma_9 & = {\mathbf{A}^d_k} (\mathbf{G}_k {\mathbf{A}^r_k}^\top + \gamma_6 {\mathbf{B}^r_k}^\top ) + \mathbf{B}^d_k  {\gamma_6}^\top {\mathbf{A}^r_k}^\top +  \gamma_7 {\mathbf{B}^r_k}^\top  \\
\gamma_{10}  &= {\mathbf{A}^r_k} (\mathbf{G}_k {\mathbf{A}^d_k}^\top + \gamma_6 {\mathbf{B}^d_k}^\top  )   +\mathbf{B}^r_k ({\gamma_6}^\top {\mathbf{A}^d_k}^\top +  \gamma_7 {\mathbf{B}^d_k}^\top )   \\
\gamma_{11} & = {\mathbf{A}^r_k} (\mathbf{G}_k {\mathbf{A}^r_k}^\top + \gamma_6 {\mathbf{B}^r_k}^\top ) + \mathbf{B}^r_k  {\gamma_6}^\top {\mathbf{A}^r_k}^\top +  {\mathbf{B}^r_k} \gamma_7 {\mathbf{B}^r_k}^\top.
\end{align*}

\subsection{Gradient of mixture likelihood} \label{Jacs}


 The gradient of the mixture log-likelihood is evaluated by utilizing properties of multivariable calculus.
  %
%
In particular, the gradients with respect to different parameters in $\phi_k$ are shown from the equations below,  
where $\gamma_7$ can be referred to in Appendix \ref{appednixB}.
The equations regarding $\Theta_1(\phi_k)$ are
\begin{align*}
    \nabla_{\mbf_k} \Tr \{{\mbSigma^o_k}^{-1} \Theta_1(\phi_k)\}
    =&  
  2\mathbf{B}^r_k \otimes {\mbSigma^r_k}^{-1}\mathbf{A}^r_k \mathbf{G}_k +  \nabla_{\mbf_k} {\mathbf{B}^r_k}  \gamma_7 {\mathbf{B}^r_k}^\top       \\
      \nabla_{\mbe_k} \Tr \{{\mbSigma^o_k}^{-1} \Theta_1(\phi_k)\}
          =& 
    2 \mathbf{B}^r_k \otimes {\mbSigma^r_k}^{-1} \mathbf{A}^r_k \hat{\mathbf{s}}_{k|T}  +  \nabla_{\mbe_k} \mathbf{B}^r_k \gamma_7 {\mathbf{B}^r_k}^\top  
   \\
      \nabla_{\mbsigma_k} \Tr \{{\mbSigma^o_k}^{-1} \Theta_1(\phi_k)\}
          =& 
     \nabla_{ \mbsigma_k}  \mathbf{B}^r_k \gamma_7 {\mathbf{B}^r_k}^\top  
   ,
\end{align*}
those for $\Theta_2(\phi_k)$
\begin{align*}
      \nabla_{\mbf_k} \Tr \{{\mbSigma^o_k}^{-1} \Theta_2(\phi_k) {\mathbf{A}^o_k}^\top\} = &   \rm{vec} ({\mathbf{M}}_{k+1|T}  {\mbSigma^d_k}^{-1}   {\mathbf{B}^d_k}^{\top} )      \\ &
      + 2 \mathbf{B}^r_k \otimes {\mbSigma^r_k}^{-1} \mathbf{A}^r_k {\mathbf{G}}_k 
       + \nabla_{\mbf_k} {\mathbf{B}^r_k}  \gamma_7 {\mathbf{B}^d_k}^\top \\
      \nabla_{\mbe_k} \Tr \{{\mbSigma^o_k}^{-1} \Theta_2(\phi_k) {\mathbf{A}^o_k}^\top\} = & \rm{vec} (\hat{\mathbf{s}}_{k+1|T}^\top  {\mbSigma^d_k}^{-1}   {\mathbf{B}^d_k}^{\top} )    \\
      & + 2 \mathbf{B}^r_k \otimes {\mbSigma^r_k}^{-1} \mathbf{A}^r_k \hat{\mathbf{s}}_{k|T} + \nabla_{\mbf_k} {\mathbf{B}^r_k}  \gamma_7 {\mathbf{B}^d_k}^\top \\
      \nabla_{\mbsigma_k} \Tr \{{\mbSigma^o_k}^{-1} \Theta_2(\phi_k) {\mathbf{A}^o_k}^\top \}
          =&  \nabla_{\mbsigma_k} \mathbf{B}^r_k \gamma_7 {\mathbf{B}^r_k}^\top  ,
\end{align*}
and those for $\Theta_3(\phi_k)$
\begin{align*}
 \nabla_{\mbf_k} & \Tr  \{ {\mbSigma^o_k}^{-1} {\mathbf{A}^o_k} \Theta_3(\phi_k) {\mathbf{A}^o_k}^\top \}=
  2 \rm{vec} (   \mathbf{G}_k {\mathbf{A}^d_k}^\top {\mbSigma^d_k}^{-1}   {\mathbf{B}^d_k}^{\top} ) \\
 &       + \nabla_{\mbf_k} \Tr \{ {\mbSigma^d_k}^{-1} \mathbf{B}^d_k  \gamma_7 {\mathbf{B}^d_k}^\top \}
      + \nabla_{\mbf_k} {\mbSigma^r_k}^{-1} {\mathbf{B}^r_k} \gamma_7 {\mathbf{B}^r_k}^\top   \\ & + 2 \mathbf{B}^r_k \otimes {\mbSigma^r_k}^{-1}\mathbf{A}^r_k {\mathbf{G}}_k \\
  \nabla_{\mbe_k}  & \Tr  \{ {\mbSigma^o_k}^{-1} {\mathbf{A}^o_k} \Theta_3(\phi_k) {\mathbf{A}^o_k}^\top   \}
 = 2 \rm{vec} ( \hat{\mathbf{s}}_{k|T}^\top {\mathbf{A}^d_k}^\top {\mbSigma^d_k}^{-1}   {\mathbf{B}^d_k}^{\top} ) \\
&  + \nabla_{\mbf_k} \Tr \{ {\mbSigma^d_k}^{-1} \mathbf{B}^d_k  \gamma_7 {\mathbf{B}^d_k}^\top\}  
 +  \nabla_{\mbe_k} {\mbSigma^r_k}^{-1} {\mathbf{B}^r_k} \gamma_7 {\mathbf{B}^r_k}^\top  \\ & + 2 \mathbf{B}^r_k \otimes   {\mbSigma^r_k}^{-1}  \mathbf{A}^r_k \hat{\mathbf{s}}_{k|T}\\      
  \nabla_{\mbsigma_k }  & \Tr \{ {\mbSigma^o_k}^{-1} {\mathbf{A}^o_k} \Theta_3(\phi_k) {\mathbf{A}^o_k}^\top \} 
      = 
   \nabla_{\mbsigma_k}\Tr \{ {\mbSigma^d_k}^{-1} {\mathbf{B}^d_k} \gamma_7 {\mathbf{B}^d_k}^\top \}  \\&
+  \nabla_{\mbsigma_k} {\mbSigma^r_k}^{-1} \mathbf{B}^r_k \gamma_7 {\mathbf{B}^r_k}^\top .
\end{align*}

\subsection{Hessian of mixture likelihoods}
The components of the Hessian of mixture log of mixture likelihood expression can be expanded and verified with equations shown below. The equations regarding $\Theta_1(\phi_k)$ are
\begin{align*}  
    \nabla^2_{\mbf_k} \Tr \{{\mbSigma^o_k}^{-1} \Theta_1(\phi_k)\} &=  \nabla^2_{\mbf_k}  {\mbSigma_k^{r}}^{-1} {\mathbf{B}^r_k}  \gamma_7 {\mathbf{B}^r_k}^\top   
    \\
      \nabla^2_{\mbe_k} \Tr \{ {\mbSigma^o_k}^{-1} \Theta_1(\phi_k)\} &=  \nabla^2_{\mbe_k}  {\mbSigma_k^{r}}^{-1} \mathbf{B}^r_k \gamma_7 {\mathbf{B}^r_k}^\top  
 \\
      \nabla^2_{\mbsigma_k} \Tr \{{\mbSigma^o_k}^{-1} \Theta_1(\phi_k)\} &=  \nabla^2_{\mbsigma_k}  {\mbSigma_k^{r}}^{-1} \mathbf{B}^r_k \gamma_7 {\mathbf{B}^r_k}^\top ,
\end{align*}
those for $\Theta_2(\phi_k)$
\begin{align*}
    \nabla^2_{\mbf_k} \Tr \{ {\mbSigma^o_k}^{-1} \Theta_2(\phi_k) {\mathbf{A}^o_k}^\top \} & =  \nabla^2_{\mbf_k} {\mbSigma^r_k}^{-1} {\mathbf{B}^r_k}  \gamma_7 {\mathbf{B}^r_k}^\top   
    \\
      \nabla^2_{\mbe_k} \Tr \{ {\mbSigma^o_k}^{-1} \Theta_2(\phi_k)  {\mathbf{A}^o_k}^\top  \}
      &=  \nabla^2_{\mbe_k} {\mbSigma^r_k}^{-1} \mathbf{B}^r_k \gamma_7 {\mathbf{B}^r_k}^\top  
   \\
      \nabla^2_{\mbsigma_k} \Tr \{ {\mbSigma^o_k}^{-1} \Theta_2(\phi_k) {\mathbf{A}^o_k}^\top\}  &=   \nabla^2_{\mbsigma_k} {\mbSigma^r_k}^{-1} \mathbf{B}^r_k \gamma_7 {\mathbf{B}^r_k}^\top  
   ,
\end{align*}
and those for $\Theta_3(\phi_k)$
\begin{align*}
      \nabla^2_{\mbe_k} \Tr \{ {\mbSigma^o_k}^{-1} {\mathbf{A}^o_k} \Theta_3(\phi_k) {\mathbf{A}^o_k}^\top \}
          =&   \nabla^2_{\mbe_k} {\mbSigma^r_k}^{-1}  {\mathbf{B}^r_k} \gamma_7 {\mathbf{B}^r_k}^\top  \\ & + \nabla^2_{\mbe_k} \Tr {\mbSigma^d_k}^{-1} {\mathbf{B}^d_k} \gamma_7 {\mathbf{B}^d_k}^\top\\
\nabla^2_{\mbsigma_k} \Tr \{ {\mbSigma^o_k}^{-1} {\mathbf{A}^o_k}  \Theta_3(\phi_k)  {\mathbf{A}^o_k}^\top \} 
          = &  \nabla^2_{\mbsigma_k} {\mbSigma^r_k}^{-1} \mathbf{B}^r_k \gamma_7 {\mathbf{B}^r_k}^\top  \\
& +  \nabla^2_{\mbsigma_k} \Tr {\mbSigma^d_k}^{-1} {\mathbf{B}^d_k} \gamma_7 {\mathbf{B}^d_k}^\top \\
      \nabla^2_{\mbf_k} \Tr \{ {\mbSigma^o_k}^{-1} {\mathbf{A}^o_k} \Theta_3(\phi_k) {\mathbf{A}^o_k}^\top  \}
          =& 
  \nabla^2_{\mbf_k} {\mbSigma^r_k}^{-1} {\mathbf{B}^r_k}  \gamma_7 {\mathbf{B}^r_k}^\top  \\ & +   \nabla^2_{\mbf_k} \Tr {\mbSigma^d_k}^{-1} {\mathbf{B}^d_k} \gamma_7 {\mathbf{B}^d_k}^\top.
\end{align*}

\subsection{Some EM lemmas} \label{EMlemmas}

Lemma~\ref{lemma:EM_proof_new} explains the lower bound maximization strategy in EM and also delineates the two main steps involved;
see, e.g., \cite{minka1998expectation}.

 \begin{lemma}  \label{lemma:EM_proof_new}
Consider $L_\phi ({\mathbb{Y}_T})$ and $\mathcal{L}(\phi,\hat{\phi}^{i})$ defined 
in \eqref{Lphi} and \eqref{Ltheta_kheta}, respectively, with 
 a  known parameter estimate $\hat{\phi^i}$.
Let \begin{align}
     l(\phi, \widetilde{p}(\mathbb{S}_{T+1}) ) =\mathbb{E}_{ \widetilde{p}(\mathbb{S}_{T+1})} 
    \log  \frac{  {p_\phi (\mathbb{S}_{T+1},\mathbb{Y}_T)} }{  \widetilde{p}(\mathbb{S}_{T+1})} \label{lphi}
\end{align}
for any distribution  $\widetilde{p}(\mathbb{S}_{T+1})$.
One has 
\begin{align*}
  L_\phi ({\mathbb{Y}_T}) \geq  l(\phi, \widetilde{p}(\mathbb{S}_{T+1}) ), 
\end{align*}
that is, $ l(\phi, \widetilde{p}(\mathbb{S}_{T+1}) ) $ is a lower bound of $ {  L_\phi ({\mathbb{Y}_T})} $.
Moreover,  let \begin{align} \label{eq:166}
\widetilde{p}(\mathbb{S}_{T+1}) = p_{ \hat{\phi}^i} (\mathbb{S}_{T+1}|\mathbb{Y}_T)  ,
\end{align}
and  denote
\begin{align}
 l(\phi,  \hat{\phi}^i) = l(\phi, p_{\hat{\phi}^i} (\mathbb{S}_{T+1}|\mathbb{Y}_T) ). \label{lphiphii}
 \end{align}
One has
\begin{align}
\hat{ \phi}^{i*}  =  \arg\max_{\phi}  l(\phi,  \hat{\phi}^i) = \arg\max_{\phi} \mathcal{L} (\phi, \hat\phi^i). \label{maxlmaxcalL}
\end{align}
\end{lemma}

Furthermore, Lemma~\ref{lemma:EM_proof} shows that, in a recursive procedure, 
any new parameter $\phi=\hat{\phi}^{i+1}$ that increases $\mathcal{L}(\phi, \hat{\phi}^i)$
from $\phi=\hat{\phi}^{i}$ also increase $L_\phi(\mathbb{Y}_T)$.

\begin{lemma} 
  \label{lemma:EM_proof}
Suppose the {parameter vector} $\hat{\phi}^{i+1}$ is produced in an iteration, 
one that 
\begin{align}
L_{\hat{\phi}^{i+1}} (\mathbb{Y}_T)    -  L_{\hat{\phi}^i} (\mathbb{Y}_T) 
\geq \mathcal{L} (\hat{\phi}^{i+1},\hat{\phi}^i) -  \mathcal{L} (\hat{\phi}^i,\hat{\phi}^i) \label{LLcalLcalL}
\end{align}
where the equality holds iff 
$p_{\hat{\phi}^{i+1}} (\mathbb{S}_{T+1}|\mathbb{Y}_T) =  p_{\hat\phi^{i}} (\mathbb{S}_{T+1}|\mathbb{Y}_T)$ . 
\end{lemma}

  Lemma~\ref{lemma-LcalL-Jac} provides a relationship between the gradients of 
 $L_\phi ({\mathbb{Y}_T})$ and $\mathcal{ L} (\phi,\hat{\phi}^i)$ evaluated at $\phi=\hat{\phi}^i$,
 called Fisher's identity \cite{douc2014nonlinear}.

 \begin{lemma} \label{lemma-LcalL-Jac} Consider $L_\phi ({\mathbb{Y}_T})$ and $\mathcal{L}(\phi,\hat{\phi}^{i})$ defined 
in \eqref{Lphi} and \eqref{Ltheta_kheta}, respectively, with 
a  known parameter estimate $\hat{\phi^i}$.
Then, 
 \begin{align} \label{JacobianLcalL}
       \frac{\partial L_\phi ({\mathbb{Y}_T})} {\partial \phi} \Big|_{\phi= \hat{\phi}^i} =   \frac{\partial \mathcal{L} (\phi, \hat{\phi}^i )} {\partial \phi} \Big|_{\phi= \hat{\phi}^i}
 \end{align}
 \end{lemma}

The property of monotonic convergence of EM undisputedly holds; see, e.g.,   \cite{gibson2005robust,wu1983convergence}. The result is summarized in the following lemma;
see, e.g., Theorem~2 of \cite{wu1983convergence}.

 {\begin{lemma}  \label{lemma-limit}
Let $\hat{\phi}^{i} \in \Phi, \; i\in 1,2, \cdots,$ be the policy parameter estimates 
 recursively generated by $\hat\phi^{i+1} = \hat{ \phi}^{i*}$ according to \eqref{maxlmaxcalL}. 
Then the limit point $\lim_{i \rightarrow \infty} \hat{\phi}^{i} = \hat{\phi}_{EM}$ exists and is a stationary point of $L_\phi(\mathbb{Y}_T)$. Also, $L_{\hat{\phi}^{i}}(\mathbb{Y}_T)$ converges monotonically to $L_{\hat{\phi}_{EM}} (\mathbb{Y}_T)$
as $i$ goes to $\infty$. 
 \end{lemma}
 
 }

\bibliographystyle{IEEEtran}
\bibliography{IEEEabrv,autosam}

\begin{IEEEbiography}[{\includegraphics[width=1in,height=1.25in,clip,keepaspectratio]{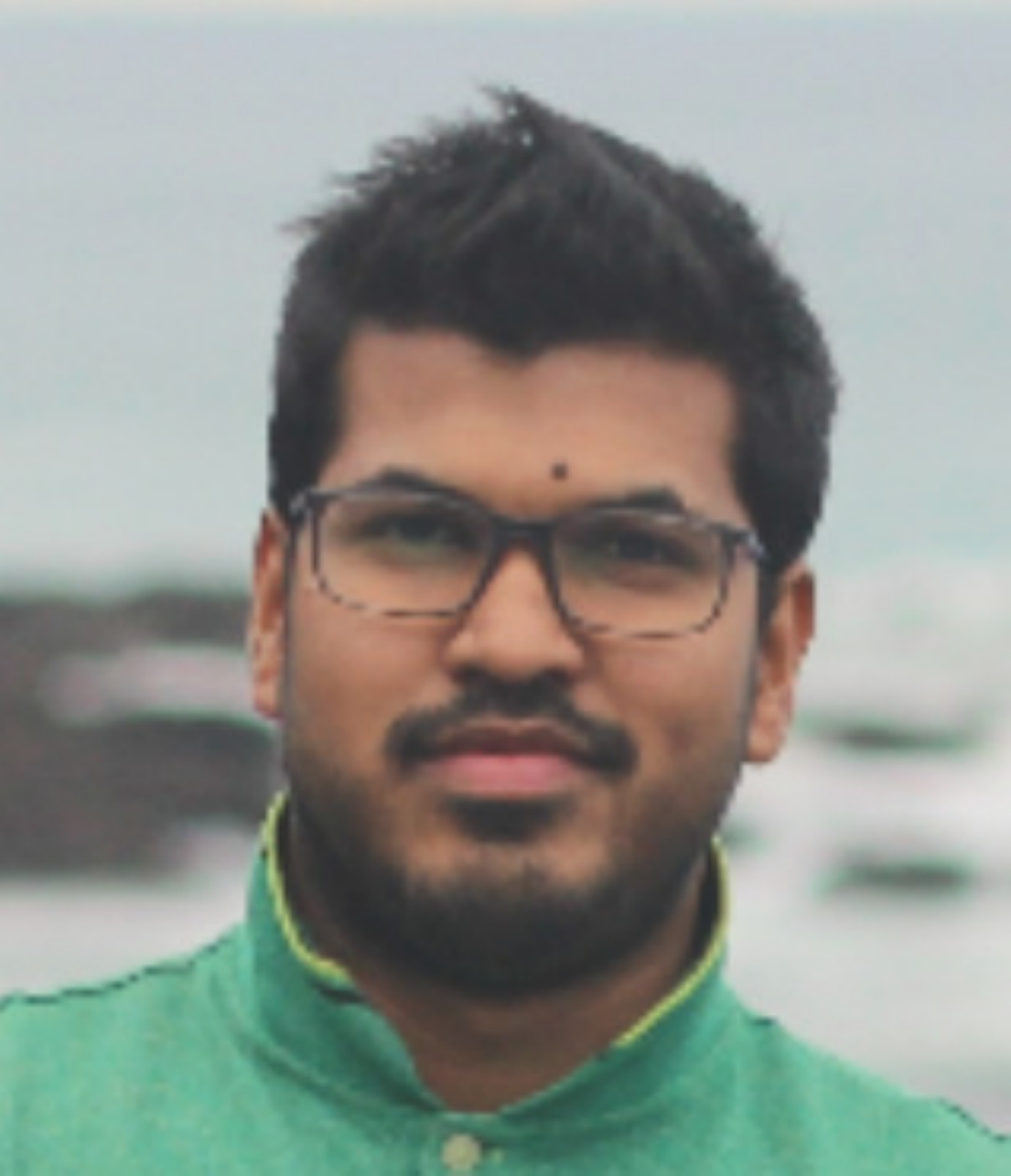}}]
{Prakash Mallick} received the B.Tech degree from National Institute of Technology, Rourkela, India and the M.Eng degree from the University of Melbourne, Australia in 2012 and 2017, respectively. He is currently a third year
Ph.D student at the University of Newcastle, Australia. His research interests include model based reinforcement learning, probabilistic inference in systems and control, and non-linear control.   \end{IEEEbiography}

\begin{IEEEbiography}[{\includegraphics[width=1in,height=1.25in,clip,keepaspectratio]{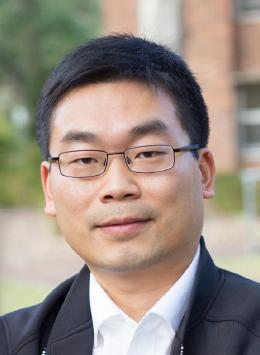}}]
{Zhiyong Chen} received the B.E. degree from the University of Science and Technology of China, and the M.Phil. and Ph.D. degrees from the Chinese University of Hong Kong, in 2000, 2002 and 2005, respectively. He worked as a Research Associate at the University of Virginia during 2005-2006. He joined the University of Newcastle, Australia, in 2006, where he is currently a Professor. He was also a Changjiang Chair Professor with Central South University, Changsha, China. His research interests include non-linear systems and control, biological systems, and multi-agent systems. He is/was an associate editor of Automatica, IEEE Transactions on Automatic Control and IEEE Transactions on Cybernetics.  \end{IEEEbiography}

\vfill



%





\end{document}